\documentclass[12pt]{article}

\usepackage[a4paper]{geometry}
\usepackage{pgfplots}
\pgfplotsset{compat=1.6}
\usepackage{sectsty}
\sectionfont{\normalsize\bfseries}
\subsectionfont{\normalsize\sf}

\usepackage [english] {babel}
\usepackage{hyperref}
\usepackage{enumitem}
\usepackage{amssymb}
\usepackage{amsmath}
\usepackage{amsthm}
\usepackage{graphicx}
\usepackage{subcaption}
\usepackage{color}
\usepackage{multirow}

\usepackage{natbib}                 
\setcitestyle{authoryear,open={(},close={)}}

\newtheorem{definition}{{\sc\bf Definition}}
\newtheorem{proposition}{{\sc\bf Proposition}}
\newtheorem{theorem}{{\sc\bf Theorem}}
\newtheorem{lemma}{{\sc\bf Lemma}}
\newtheorem{example}{{\sc\bf Example}}
\newtheorem{remark}{{\sc\bf Remark}}

\newtheorem{assumption}[theorem]{Assumption}

\def\N{{\mathbb N}}

\def\Z{{\mathbb Z}}
\def\P{{\mathbb P}}
\def\E{{\mathbb E}}
\def\R{{\mathbb R}}

\def\P{\mathbb{P}}

\def\cals_+{{\cals_+}}

\def\calh{{\mathcal{H}}}
\def\call{{\mathcal{L}}}
\def\cals{{\mathcal{S}}}

\newcommand{\var}{{\rm Var}}
\newcommand{\cov}{{\rm Cov}}
\newcommand{\dd}{\mathrm{d}}
\newcommand{\wh}{\widehat}

\newcommand{\stas}{\stackrel{\rm a.s.}{\rightarrow}}

\allowdisplaybreaks 

\usepackage{algorithm}
\usepackage[noend]{algpseudocode}
\makeatletter
\def\BState{\State\hskip-\ALG@thistlm}
\makeatother

\title{Variable selection for the prediction of C[0,1]-valued AR processes using RKHS}
\author{
Beatriz Bueno-Larraz\footnote{Universidad Aut\'onoma de Madrid, Departamento de Matem\'aticas, Spain, email: \tt{beatriz.bueno@uam.es} \& \tt{beatriz.bueno.larraz@gmail.com}}\;\footnote{Corresponding author}
\and Johannes Klepsch\footnote{Center for Mathematical Sciences, Technische Universit\"at M\"unchen,  85748 Garching, Boltzmannstra{\ss}e 3, Germany, email: \tt{j.j.klepsch@gmail.com}}
}\date{March 2018}

\begin{document}

\maketitle

\begin{abstract}
\setlength{\baselineskip}{1.8em}
A model for the prediction of functional time series is introduced, where observations are assumed to be continuous random functions. We model the dependence of the data with a nonstandard autoregressive structure, motivated in terms of the Reproducing Kernel Hilbert Space (RKHS) generated by the auto-covariance function of the data. The new approach helps to find relevant points of the curves in terms of prediction accuracy. This dimension reduction technique is particularly useful for applications, since the results are usually directly interpretable in terms of the original curves. An empirical study involving real and simulated data is included, which generates competitive results. Supplementary material includes R-Code, tables and mathematical comments.
\end{abstract}

\

\section{Introduction}

Functional data analysis (FDA) is one of the  answers to the recent rise of complex data. It consists in viewing observations as entire curves instead of individual data points. 
For a recent overview of the literature of FDA we refer to \citet{cuevas2014}. 

\

\textit{Functional time series and variable selection}

In many applications the data consists of a single curve $z(s)$ sequentially recorded in time, up to an instant point $s\in (-\infty,N]$ (or $[0,N])$. Typically these curves present a periodical behavior. Common examples are daily financial or meteorological records. This directly suggests to split the curve $z$ into pieces of the same length, sharing a common structure. These new curves are denoted as $x_n\equiv x_n(\cdot)$, for $n\leq N$, and are usually rescaled to the interval $[0,1]$ in order to make them independent of the time unit of measure. Each of these curves $x_n$ are understood to be a realization of the corresponding random process $X_n$, which is a random variable whose samples are functions instead of real numbers. The set of these processes $X_n$, indexed by $n\in\Z$, is known as \textit{functional time series} (see \cite{alvarez2017} for details). In this work we focus on predicting the curve $x_n$ given the previous ones.

In practice, these curves are usually recorded as high-dimensional vectors with highly correlated entrances. Then, the need of dimension reduction techniques that take into account the continuous nature of the data arises. We propose to replace the whole curves $x_n(\cdot)$ with the $p$ most relevant evaluations $x_n(t_1),\ldots, x_n(t_p)$ for the prediction of $x_{n+1}(\cdot)$ (in a similar sense as in \citet{kargin} and \citet{mokhtari2003}). This is equivalent to directly select the points $t_1,\ldots,t_p$ in $[0,1]$ under a suitable optimality criterion. Although we end up with a finite dimensional vector, the problem is fully functional, since the definition of this criterion is based on the whole curves. This technique is commonly know as \textit{variable selection}. Its main advantage in comparison with other dimension reduction techniques is the interpretability in terms of the original data, which is usually desired in real applications. For instance, one of the real data sets tested is the daily indoor temperature of a house which uses solar energy. With our technique we detect that one of the most relevant points to predict the temperature one day is 8:00 of the previous morning, which is not evident at first sight.

In the present paper we introduce a predictor based on a functional autoregressive (AR) model which will be especially suitable for variable selection purposes. \citet{bosq} gives a good introduction into the field of linear processes in function spaces and introduces functional autoregressive processes in depth. AR models have shown to be a valuable tool in functional time series analysis since they combine computational tractability with sufficient generality (e.g. \cite{bessecardot,kargin,didericksen,aue}). In particular, we adapt the methodology introduced by \citet{berrendero2018} for the problem of scalar regression for independent data. We define a new functional autoregressive model based on the Reproducing Kernel Hilbert Space (RKHS) generated by the auto-covariance function of the time series. 

\

\textit{Classical approaches}

The standard assumption in the literature is that $L^2[0,1]$, the space of square-integrable functions on $[0,1]$, is the space into which the curves fall. This is a sensible choice, since $L^2[0,1]$ is a separable Hilbert space and offers desirable geometric properties through the definition of the natural scalar product. However, considering our variable selection purpose, the main drawback of $L^2[0,1]$ is that, strictly speaking, it consists of equivalence classes of functions. That is, two functions represent the same $L^2$-function if the set where they differ has measure zero. In other words, for any particular point $s\in[0,1]$, the value $f(s)$ is not well-defined for any function $f$ in this space. 

Since we require pointwise evaluations $x_n(s)$ of the curves, the space of continuous functions on $[0,1]$, $C[0,1]$, which is a Banach space with the supremum-norm, is a more natural space to work in. In addition, the subsequent change of norm allows us to obtain uniform convergence results. The problem of estimating AR models in $C[0,1]$ has been already addressed in the literature (e.g. \cite{ruiz2018} and references therein). The usual methodology is to project the curves onto a finite dimensional subspace of $C[0,1]$, spanned by some eigenfunctions of the covariance operator of the data, as in \citet{pumo1998}. Some limitations of this principal component approach have been discussed extensively in the literature. For instance, the resulting space is shown to be optimal in order to represent the variability of the process, but the dependence  might be lost by the dimension reduction (\citet{kargin} and \citet{hoermann}). Furthermore, \citet{bernard} indicates the sensitivity of the proposal to small errors in the estimation of small eigenvalues.

\

\textit{Some relevant results of this work}

In this paper the projection on a finite dimensional space is replaced by the choice of the evaluations $x_n(t_1),\ldots, x_n(t_p)$. 
In the cases relevant for variable selection, it is shown that the new model falls into the wide class of Banach space--valued processes (ARB$(q)$) introduced in \citet{bosq}, which directly gives us sufficient conditions for the existence of a unique stationary solution. Notably, for instance the well-known Ornstein-Uhlenbeck process satisfies our model. 
In these cases we are able to prove, under some standard conditions, almost sure convergence of the estimated points and also of the estimated curves, both uniformly and in $L^2[0,1]$. In this setting our predictor coincides with the optimal one, in the sense that it is the best probabilistic predictor, as stated in \citet{mokhtari2003}. In addition, we develop a consistent estimator for the number $p$ of relevant variables to select.  

The advantages of predicting autoregressive processes with this new approach are numerous, besides the ones already mentioned. From an applied perspective, the proposed method is flexible concerning the structure of the data: whether it is recorded on a grid or available as continuous functions - the methodology remains similar with slight technical differences. Nevertheless, for theoretical considerations the data is assumed to be given in a fully functional fashion. Besides, the use of this RKHS based model avoids the need of inverting the covariance operator since this is carried out, in some sense, by the inner product of the space. 

In order to show the practical relevance of the method, a simulation study is conducted. To evaluate the performance in the real world four real data sets are studied. The execution times of the tested methods are also analyzed. Our proposal is competitive both in prediction accuracy and computational efficiency.

\

\textit{Related literature}

The literature in the field of functional time series analysis is developing quickly.  Recent publications include time-domain methods like \citet{weaklydep}, where a weak dependence concept is introduced, \citet{aue}, \citet{kk} and \citet{KKW}, where prediction methodologies based on linear models are developed, and \citet{auekle}, where an estimator of functional linear processes based on moving average model fitting is derived. Besides, another examples of statistical papers taking advantage of the usefulness of Reproducing Kernel Hilbert Spaces are, among others, \citet{hsing, kad16, berrendero2017class, berrendero2018} and \cite{yang2018}. 

Some interesting variable selection techniques for functional regression with independent observations are \cite{aneiros2014, delaigle2012, mckeague2010, ferraty2010} and \cite[Sec. 4.3]{shi2011}. See also Section 7.3 of \cite{cuevas2014} for a brief review. There are also some proposals for feature selection on multivariate time series, like \cite{fan2010,lam2012,liu2014,tran2015} and the references therein. However, as far as we know, there are no previous approaches to variable selection for continuous time series in the same sense as it is presented here.

\

\textit{Organization of the document}

The paper is organised as follows. After introducing the notation and some background on RKHS theory, we define in Section \ref{Sec:VS} the new functional autoregressive model for variable selection. In Appendix A the mathematical properties of such model are carefully studied. The asymptotic properties of the sample estimators are shown in Section \ref{Sec:Asymp}. In this section the estimator for the number $p$ of relevant variables is presented. Sections \ref{Sec:ExpSetting} and \ref{Sec:Experiments} include the experimental study along with some practical consideration for the implementation. Some proofs, which are mainly based on the theory developed in \cite{berrendero2018}, are included in Appendix B. Appendix C includes some tables and pseudo-codes.


\section{Methodology}


\subsection{Notation}

Our object of interest is a random curve observed over $(-\infty,N]$, for $N\in\Z$ (or over $[0,N]$). In order to better handle its behavior when $N$ increases, the whole curve is usually split into intervals of the same length. Then, each piece is a function rescaled to $[0,1]$, randomly generated from a stochastic process $X_n$, for $n\in \Z$ and $n\leq N$. A stochastic process can be thought as a random variable whose samples, denoted as $x_n$, are functions. Such a set of stochastic processes is known as a functional time series.

We assume that the curves inside each interval of length one are continuous, that is, each $X_n$ takes values in $C[0,1]$ (the space of continuous functions over $[0,1]$). We denote by $\Vert \cdot \Vert$ the supremum norm in this space, 
$$\Vert f \Vert = \sup_{s\in[0,1]}|f(s)|, \ \text{ for } f\in C[0,1].$$ 

A standard assumption is that the random variable $\sup_{s\in[0,1]}|X_n(s)|$ has finite variance. In this case each evaluation $X_n(s)$, for $s\in [0,1]$, also has finite variance. A functional time series satisfying this condition is said to be stationary if its mean functions $\E[X_n]$ do not change with $n$ and $\cov\big(X_{n+r}(s),X_n(t)\big)$ equals $\cov\big(X_r(s),X_0(t)\big)$, for every $s,t\in[0,1]$, $n,r\in \{0\}\cup\N$. Under stationarity the lagged covariance function, $\cov\big(X_r(s),X_0(t)\big)$, is denoted as $c_r(s,t)$. For the sake of clarity in the equations, and since the time series throughout this work are stationary, we make the following abuse of notation: we assume to work with the centered processes $X_n-\E[ X_n]$, denoted simply by $X_n$. 

We denote the vectors of points as $T_p=(t_1,\ldots,t_p)\in [0,1]^p$, where the subindex indicates the dimension. The covariance matrix of the random variables $X_n(t_1),\ldots,X_n(t_p)$ indexed by $T_p$ is $\Sigma_{T_p}$. Moreover, for a general function $f:[0,1]\to\R$, the evaluation $f(T_p)$ is understood to be the column vector with entries $f(t_j)$, $t_j\in T_p$. Similarly for functions in two variables.

As usual in statistics, we use a hat to denote the estimations derived from the samples. In addition, an asterisk will indicate the optimal quantities under some criterion.


\subsection{Some background on Reproducing Kernel Hilbert Spaces}\label{Sec:RKHS}

The following discussion, although interesting, is not directly needed in order to apply the proposed method in practice. Therefore, a more applied reader could directly skip to Section \ref{Sec:VS} if desired.

Let $X(\cdot)\equiv X_n(\cdot)$ be a centered stochastic process in $[0,1]$ such that $X(s)$ has finite variance for all $s\in[0,1]$. Being $c_0(s,t)$ the auto-covariance function (or covariance kernel) of the process, we can define the auxiliary space
\begin{equation}
\calh_0(X):=\big\{f\in L^2[0,1] \ : \ f(\cdot)=\sum_{i=1}^p a_ic_0(t_i,\cdot),\ a_i\in{\mathbb R},\ t_i\in[0,1],\ p\in{\mathbb N}\big\}
\end{equation}
with inner product
$
\langle f,g\rangle_{\calh_0}=\sum_{i,j}a_i b_j c_0(t_i,s_j),
$
where $f(\cdot)=\sum_i a_i c_0(t_i,\cdot)$ and $g(\cdot)=\sum_j b_jc_0(s_j,\cdot)$. The RKHS associated with $c_0$, $\calh(X)$, is the completion of this space; all the functions of $\calh_0(X)$ plus their limits when $p\to\infty$ with respect to the norm $\Vert\cdot\Vert_{\calh_0}$. The norm in the RKHS is the continuous extension of the norm in $\calh_0(X)$ and it is denoted as $\Vert \cdot \Vert_\calh$. Thus, $\Vert f\Vert_\calh$ coincides with $\Vert f\Vert_{\calh_0}$ for all function $f\in\calh_0(X)$.


Reproducing Kernel Hilbert Spaces appear occasionally in the literature of functional data and machine learning, since they are useful to impose smoothness conditions and help to reduce noise and irrelevant information. However we use them with a different goal, since our interest lies in variable selection. Therefore, the so-called \textit{reproducing property} of these spaces is particularly useful. It states that, for all $f\in \calh(X)$ and $s\in[0,1]$, \ $\langle f, c_0(s,\cdot)\rangle_\calh = f(s)$. That is, the auto-covariance function behaves, in some sense, like Dirac's delta. We are interested in selecting variables on the trajectories drawn from the process $X(\cdot)$. However, in general, the realizations of the process do not belong to $\calh(X)$ with probability one (e.g., Theorem 11 of \cite{pillai2007}). Thus, this reproducing property can not be directly  applied to the trajectories. 

In order to circumvent this issue, we use another Hilbert space also closely related to the process,
$$\call_0 (X) \ = \ \big\{ U : \ U = \sum_{i=1}^p a_i X(t_i), \ a_i\in\R, \ t_i\in [0,1], \ p \in \N \big\}.$$
As before, $\call(X)$ is the closure of this space. All the random variables in this space have finite variance, which coincides with their norms.

\begin{remark}\label{Rem:Dense}
By definition of both spaces, the finite sums $\sum_{i=1}^pa_ic_0(t_i,\cdot)$ are dense in $\calh(X)$ and the finite sums $\sum_{i=1}^pa_iX(t_i)$ are dense in $\call(X)$, with their corresponding norms (i.e. every function of these spaces can be arbitrarily well approximated with a finite sum).
\end{remark}

These two spaces $\call (X)$ and $\calh(X)$ can be connected using the following congruence (bijective transformation preserving the inner product), named \textit{Lo\`eve's isometry} (\citet[Theorem 35]{berlinet2004} and \citet[Lemma 1.1]{lukic2001})
\begin{equation}\label{Eq:Isom}
    \begin{array}{c c c l}
\Psi_X: & \call(X) &\to& \calh(X) \\
& U &\mapsto& \E [U X(\cdot)].
\end{array}
\end{equation}
If the random variable $U$ is an element of $\call_0(X)$, i.e. $U = \sum_{i=1}^p a_i X(t_i)$, its image by the isometry is given by
\begin{align}
    \big(\Psi_X(U)\big)(\cdot) \ = \ \E \Big[ X(\cdot)\sum_{i=1}^p a_i X(t_i)\Big] \ = \ \sum_{i=1}^p a_i \E [ X(\cdot) X(t_i) ] \ = \ \sum_{i=1}^p a_i c_0(t_i,\cdot) \label{eq:isomfin}.
\end{align}
Therefore, when applying the inverse of $\Psi_X$ to a function in $\calh_0(X)$ we obtain a finite combination of evaluations $X(t_i)$. That is, replacing the inner product in $\calh(X)$ with $\Psi_X^{-1}$ we recover the Dirac's delta behavior for the trajectories, since $\Psi_X^{-1}(c_0(s,\cdot))=X(s)$. Next we see how to use this isometry to perform variable selection in AR processes.


\section{Model definition and variable selection}\label{Sec:VS}

Given $x_n\in C[0,1], n\in\Z,$ trajectories drawn from a functional time series $X_n$, the standard autoregressive model is of the form (see Chapter~6 \citet{bosq})
\begin{align}
    x_n = \rho (x_{n-1}) + \varepsilon_n, \qquad n\in\Z, \label{ARB}
\end{align}
for some white noise process $\varepsilon_{n}, n\in\Z,$ in $C[0,1]$ and some bounded linear operator $\rho$ (i.e. $\sup_f\Vert \rho(f)\Vert$ finite for $f\in C[0,1]$ and $\Vert f \Vert \leq 1$). In this section we propose a particular functional AR model ``customized'' to give a well-founded framework for variable selection.


We work with a centered stationary time series $X_n$ in $C[0,1]$ such that $\E\big[\big(\sup_s |X_n(s)|\big)^2\big]$ is finite. In order to perform variable selection we propose to define the model
\begin{equation}\label{Eq:SparseX}
 X_n(\cdot) \ = \ \sum_{j=1}^p \alpha_j(\cdot) X_{n-1}(t_j) + \varepsilon_n(\cdot),  
\end{equation}
where $\alpha_j(\cdot)$ are continuous functions in $[0,1]$ and $\varepsilon_n$ is a strong $C[0,1]$-white noise point-wisely uncorrelated with $X_n$. That is, all the curves depend on the same set of points $t_1,\ldots,t_p$, regardless of the index $n$. Comparing this expression with the standard AR model of Equation \eqref{ARB}, we are using
$\rho (f)(\cdot) = \sum_{j=1}^p \alpha_j(\cdot) f(t_j)$, for $f\in C[0,1]$.
The following example, taken from \cite{bosq}, shows that this model class has good amount of generality.
\begin{example}\label{Ex:OU}
Let $Z$ be a continuous version of the Ornstein-Uhlenbeck process, 
$$Z(s) = \int_{-\infty}^s e^{-\theta(s-t)}\dd B(t),$$
where $B$ is a standard Brownian motion. If we define $X_n(s) = Z(n+s)$ for $s\in[0,1]$, $X_n(s)$ can be rewritten as $e^{-\theta s}X_{n-1}(1) +  \varepsilon_n(s)$, where now $\varepsilon_n$ is a white noise given by $\int_n^{n+s} e^{-\theta(n + s - t)}\dd B(t)$. That is, Equation \eqref{Eq:SparseX} is fulfilled with $p=1$, $t_1 = 1$ and $\alpha_1(s) = e^{-\theta s}$. 
\end{example}

This is reasonable, since O.U. processes are Markov, which means that the probability events involving $X_n(\cdot)$ only depend on $X_{n-1}(1)$. For real data sets one can not ensure whether model \eqref{Eq:SparseX} holds, but the results in Section \ref{Sec:Experiments} reveal that it is, at least, a good approximation of reality. In any case, it is important to emphasize that in order to carry out the variable selection we do not assume that the process fits model \eqref{Eq:SparseX}. We merely search the points $t_1,\ldots,t_p$ such that an expression as in \eqref{Eq:SparseX} approximates the real process $X_n$ according to some criterion. 
 
Then, we propose the following definition.
\begin{definition}\label{def:fcar1}
A functional time series $X_n$ such that $\E \big[(\sup_s |X_n(s)|)^2\big]<\infty$, $n\in\Z$, is called Functional Continuous Autoregressive process of order 1 (FCAR(1)) if it is stationary and  it can be expressed as in Equation \eqref{Eq:SparseX}.
\end{definition}

As pointed out in Appendix A, a sufficient condition for a FCAR(1) process to have a unique strictly stationary solution is that $\sum_{i=1}^p \Vert \alpha_i \Vert < 1$, which is an easily verifiable condition. 

\begin{assumption} \label{ass:1}
$X_n, n\in\Z$ is a FCAR process such that $\sum_{i=1}^p \Vert \alpha_i \Vert < 1$. 
\end{assumption}

An extension of model FCAR(1) to FCAR(q) can be carried out whenever $X_n(s) = Z(s+n)$ for $s\in [0,1]$ and $Z$ is a stationary process with continuous trajectories. All the theory presented in the paper remains valid in this case, with some intuitive additional assumptions on the model parameters. For the sake of clarity we restrict ourselves to $q=1$. Nevertheless, we include some comments along the document to clarify the changes due to this extension. One of the necessary generalizations is to split the process as $Z_{n,q}(s) = Z(s+n-q+1)$, $s\in [0,q]$.  Now $\call (Z_{n,q})$ is the space generated by the random variables $X_n(s),\ldots,X_{n-q}(s)$ and $\calh (Z_{n,q})$ is the RKHS associated with this lagged process, whose reproducing kernel $c_0(s,t)$ for $s,t\in[0,q]$ is the auto-covariance function of $Z_{n,q}$. Thus, Equation \eqref{Eq:SparseX} is rewritten as
$$X_n(\cdot) \ = \ \sum_{j=1}^{p^{(1)}} \alpha_j^{(1)}(\cdot) \, X_{n-1}\big(t_j^{(1)}-q+1\big)  + \ldots +  \sum_{j=1}^{p^{(q)}} \alpha_j^{(q)}(\cdot) \, X_{n-q}\big(t_j^{(q)}\big).$$ 

An extension of model \eqref{Eq:SparseX} is analyzed hereunder, making clear that this approach includes a wide class of processes, and not only the finite approximations of Equation \eqref{Eq:SparseX}. 

Model \eqref{Eq:SparseX} is a particular case of the more general (and fully functional) model,
\begin{align}
X_n(\cdot) \ = \ \Psi_{X_{n-1}}^{-1}(\phi(\cdot,\star)) + \varepsilon_n(\cdot),  \qquad n\in\Z,  \label{fullmodel}
\end{align}
where $\phi\in\calh(X)$ and $\Psi_{X_{n-1}}^{-1}$ is the inverse of the Lo\`eve's isometry defined in Equation \eqref{Eq:Isom}. This more general formulation is useful when trying to prove asymptotic results on the model, although it is not directly used in practice. As it is sometimes the case, a slight generalization helps to make theory simpler. Another benefit of changing the work space to $\calh(X)$ is that finding a solution of this model does not require to invert the covariance operator. It could be understood as if the ``inversion'' was intrinsically carried out by the Lo\`eve's isometry. By means of Equation \eqref{eq:isomfin}, it is easy to see that model \eqref{Eq:SparseX} is recovered for 
\begin{equation}\label{Eq:SparseGeneral}
    \phi(s,\cdot) \ = \ \sum_{j=1}^{p} \alpha_j(s) c_0(t_j,\cdot) \ \in \ \calh (X).
\end{equation}
In fact, both models are almost equivalent since, as pointed out in Remark \ref{Rem:Dense}, these finite linear combinations are dense in $\calh(X)$ (at the end of Appendix A it is included a convergence result in this regard). Therefore, any possible function $\phi(s,\cdot)$ in \eqref{fullmodel} can be approximated arbitrarily well by just increasing the number of points $p$. Besides, in most cases $\calh(X)$ is a dense subspace of $L^2[0,1]$. The practical examples displayed in Section \ref{Sec:Experiments} demonstrate that a small number $p$ is usually enough to obtain good approximations. These experimental results also support the claim that any function can be approximated as in \eqref{Eq:SparseGeneral}, since the method performs well even when the model is not satisfied.


\subsection{Optimality criteria}

We now focus on the optimality criterion for variable selection. The main objective of the regression method is to predict $X_n(\cdot)$ given $X_{n-1}(\cdot)$. Since each $X_n(s)$ for $s\in[0,1]$ has finite variance, the first approach could be to minimize the variances of the residuals
$$q(T_p\ ; \ \alpha_1,\ldots, \alpha_p) (s) \ = \ \E\Big[\Big( X_n(s) - \sum_{j=1}^p \alpha_j(s)X_{n-1}(t_j) \Big)^2\Big],$$
in the same spirit as in \citet{berrendero2018}, where the coefficients $\alpha_j(s)$ (which are just real numbers) depend on the points $t_1,\ldots,t_p$. Due to the stationarity of the process, $q$ does not change with $n$. 

In model \eqref{Eq:SparseX} all processes $X_n(\cdot)$ depend on the same set of points  $t_1,\ldots,t_p$, independently of $s$. Therefore, we have to find the functions $\alpha_j(\cdot)$ such that for each $s\in[0,1]$ the evaluations $\alpha_j(s)$ give the best approximation of $X_n(s)$ for a given set of points $T_p$. Then, integrating $q^2$ over $s$ leads to
\begin{equation}\label{Q1}
Q(T_p) \ := \  \int_0^1 \min_{\alpha_j(s)\in\R} q(T_p\ ; \ \alpha_1,\ldots, \alpha_p)^2 (s) \ \mathrm{d}s.
\end{equation}
This function $Q$ can now be minimized with respect to $T_p$. From \cite{berrendero2018} we know that $q^2(s)$ is a convex function in $\alpha_j(s)$ for each $s\in[0,1]$. Thus, we can obtain an explicit expression of the minimizing functions, denoted by $\alpha_j^*(s)$, pointwise for each $s\in[0,1]$. We see in the proof of Proposition \ref{Prop:Equiv} that these optimal functions are given by 
$$(\alpha_1^*(s),\ldots,\alpha_p^*(s)) \ = \ \Sigma_{T_p}^{-1} \ c_1(s,T_p),$$
where $c_1(\cdot,T_p) = (c_1(\cdot,t_1),\ldots,c_1(\cdot,t_p))'$ is the vector of lagged-covariances. 

Up to now we have assumed that all points $t_j$ belong to $[0,1]$. However, if we want to have identifiability of the set $T_p \in [0,1]^p$, we need to restrict the search to a compact subset of this space. This problem has been solved in different ways in the literature. The chosen solution is to work, for some $\delta>0$, in
$$\Theta_p = \{ T_p=(t_1,\ldots,t_p)\in [0,1]^p \ : \ t_{i+1}-t_i\geq \delta, \ \text{for} \ i=1,\ldots,p \},$$
which is the space proposed in \cite{berrendero2018}. Other possibilities can be studied, for instance the space used by \citet{muller2016}. The choice of a value $\delta>0$ is mainly technical, to avoid problems with the invertibility of the covariance matrices. In addition it allows us to obtain meaningful sets $T_p$, since it discards reoccuring points. It is also not a strong restriction in practice, since usually the data is given in a discretized fashion, and the value $\delta$ can be chosen as small as desired. However, as pointed out in \cite{berrendero2018}, all the theory remains valid for $\delta=0$, by simply adjusting everywhere the value $p$ to the dimension of the vector without repeated entries.

Although the optimality criterion defined by $Q$ is theoretically sound, it lacks an easily computable expression. We see in the following result that it can be rewritten in a more feasible way. 

\begin{proposition}\label{Prop:Equiv}
Let $X_n, n\in\Z$ be a FCAR(1) process satisfying Assumption~\ref{ass:1} and with $\E\Vert\varepsilon_n^2\Vert <\infty$. 
Then, 
$$\arg\min_{T_p\in\Theta_p} Q(T_p) \ = \ \arg\max_{T_p\in\Theta_p} Q^0(T_p),$$
where $Q^0(T_p) := \int_0^1 c_1(s,T_p)' \Sigma_{T_p}^{-1} c_1(s,T_p) \dd s$.
\end{proposition}

The proof of this result is an extension to the current setting of the proof of Proposition 1 in \cite{berrendero2018} and can be found in Appendix B. 

For the criterion of the $FCAR(q)$ with $q>1$ we have to substitute in this proposition the function $c_1(s,t)$ by the continuous picewise-defined function
\begin{equation}\label{Eq:covZq}
    c(s,t) \ = \ c_i(s,t-q+i) \ \text{ for } \ t\in (q-i,q-i+1] \ \text{ and } \ s\in [0,1].
\end{equation}
We need the additional assumption that the random variables $Z_{n,q}(s)$ for each $s$ are all linearly independent, to ensure the invertibility of the covariance matrices of $Z_{n,q}$ evaluated in $(t_1^{(1)},\ldots,t_{p^{(q)}}^{(q)}) \in [0,q]^p$. This assumption introduces some further restrictions to the model. For instance, the functions $\alpha_j^{(i)}(s)$ can not vanish for $s\in [0,1]$ and $1<i\leq q$.


\subsection{Estimation from the sample}

As mentioned, the optimality criterion defined by $Q^0$ is simple to implement in practice. In this section we study the asymptotic properties of the natural estimator of $Q^0$. These results are useful when studying the asymptotic behavior of the selected points and the estimated trajectories in the next section. We work with a sample $x_1,\ldots,x_m$ of size $m$ drawn from a FCAR(1) process satisfying Assumption \ref{ass:1}. Then, for a given number of points $p$, the natural estimator for the functions $Q^0(T_p)$ is 
\begin{equation}\label{Eq:Qn}
    \wh Q^0_m(T_p) = \int_0^1 \widehat{c}_1(s,T_p)' \ \widehat{\Sigma}_{T_p}^{-1} \ \widehat{c}_1(s,T_p) \mathrm{d}s,
\end{equation}
where $\widehat{c}_1(\cdot,T_p)' = (\widehat{c}_1(\cdot,t_1), \ldots, \widehat{c}_1(\cdot,t_p))$ and $\wh c_1$ is the usual estimator of the covariance function
$$\widehat{c}_1(s,t_j) \ = \ \frac{1}{m-1} \sum_{i=1}^{m-1} x_{i+1}(s)x_i(t_j). $$
The entries of the sample covariance matrix, $\wh c_0(t_i,t_j)$, $t_i,t_j\in T_p$, are computed equivalently. According to this criterion, we propose to select as the most relevant points
\begin{equation}\label{EstT}
\widehat{T}_{p,m} \ = \ \arg\max_{T_p\in\Theta_p} \widehat{Q}_{0,m}(T_p).
\end{equation}

In Section \ref{Sec:Asymp} we prove some consistence results for this estimator, under the assumption that the finite dimensional model defined by Equation \eqref{Eq:SparseX} holds. To this end, we first need convergence results of the sample covariances involved in the expression of $\wh Q^0_m$. The main one is based on a result of \citet{pumo1998}.

\begin{lemma}\label{Lemma:UnifConv}
Assume that $X_n, n\in\Z$ is a FCAR(1) process satisfying Assumption~\ref{ass:1} and:
\begin{itemize}
    \item[H1.] The process $X_n(t)X_n(s)$ for $t,s\in[0,1]$ is uniformly geometrically strong mixing.
    \item[H2.] (Cramer conditions) For every $t,s\in [0,1]$ there exist $d>0$ and $D<\infty$ such that     
    \begin{eqnarray*}
    & \bullet& d \ \leq \ \E [X_0^2(t)X_0^2(s)] \ \leq \ D \hspace*{1cm} \text{and} \hspace*{7cm}\\[2mm]
     &\bullet& \E | X_0(t)X_0(s) |^k \ \leq \ D^{k-2}k! \ \E [X_0^2(t)X_0^2(s)] \ \text{ for } \ k \geq 3.
    \end{eqnarray*}
\end{itemize}
Then,
\begin{align}
    &\sup_{t,s\in[0,1]} \vert {\wh c}_0 (s,t) - c_0 (s,t) \vert \stas 0 \qquad \text{and} \label{UnifConv}\\
    &\sup_{t,s\in[0,1]} \vert {\wh c}_1 (s,t) - c_1 (s,t) \vert \stas 0. \label{UnifConvC1}
\end{align}
\end{lemma}

\begin{proof}
By Lemma 1 of \cite{pumo1998} we know that for some positive constants $A_1, A_2, A_3$,
\begin{eqnarray*}
\P\left( \sup_{t,s\in[0,1]} \vert {\wh c}_0(s,t) - c_0(s,t) \vert \geq  \varepsilon\right)  & \leq & (2\sqrt{m}+A_1) \ {\rm exp} \left( -A_2 \varepsilon^2\sqrt{m}\right) \\ && + A_3 \varepsilon^{\frac{2}{5}} m \ {\rm exp}\left( -{\rm log}(r^{-1}) \sqrt{m} \right),
\end{eqnarray*}
where $m$ is the sample size and $0<r<1$ is given by assumption H1. By Borel-Cantelli, if the sums over $m$ of these probabilities are finite for every $\varepsilon>0$, we get the almost sure convergence stated in Equation \eqref{UnifConv}. The sum is of order
$$\sum_{m=1}^\infty \P\left( \sup_{t,s\in[0,1]} \vert {\wh c}_0(s,t) - c_0(s,t) \vert \geq \varepsilon\right) \ \sim \  2 \sum_{m=1}^\infty \frac{\sqrt{m}}{{\rm e}^{C_\varepsilon\sqrt{m}}} + \sum_{m=1}^\infty \frac{A_1}{{\rm e}^{C_\varepsilon\sqrt{m}}} + D_\varepsilon \sum_{m=1}^\infty \frac{m}{{\rm e}^{C_r\sqrt{m}}},$$
where $C_\varepsilon, C_r, D_\varepsilon > 0$ and these three series converge, for example by the limit comparison test with $\sum m^{-\gamma}, \gamma>1$. Concerning \eqref{UnifConvC1}, the same Lemma 1 of \cite{pumo1998} states that the bounds for these probabilities are equivalent but with $m-1$ in place of $m$.
\end{proof}

These Cramer conditions appear often in the literature related with limit theorems for AR processes in Banach spaces. For instance, all bounded processes satisfy them, and also the Ornstein-Uhlenbeck process of Example \ref{Ex:OU}. In the latter case, $|X_n(s)X_n(t)| = e^{-k(t+s)}X_{n-1}(1)^{2}$, then,
$e^{-k(t+s)}\E X_0(1)^{2k} \ \leq \ D^{k-2}k! e^{-2(t+s)}\E X_0(1)^4,$
where $X_0(1)$ follows a $\mathcal{N}(0,0.5)$. Using the expression for the moments of a Gaussian variable, 
$$e^{-k(t+s)}\frac{(2k)!}{2^{2k}k!} \ \leq \ D^{k-2}e^{-2(t+s)}\frac{3k!}{4},$$
which is satisfied, for instance, for $D\geq 5e^{-2}/12$.

Given the previous result, we show the uniform convergence of the sample criterion function to its population counterpart and that both functions are continuous on $\Theta_p$.

\begin{lemma}\label{Lemma:ContConv}
 Assume that $X_n$ satisfies the same hypotheses as in Lemma \ref{Lemma:UnifConv}. Let $p\geq 1$ be such that the covariance matrices $\Sigma_{T_p}$ are invertible for all $T_p\in\Theta_p$. Then $Q^0$ and $\widehat{Q}_{0,m}$ are continuous on $\Theta_p$ and \
$\sup_{T_p\in\Theta_p} \vert \widehat{Q}_m^0(T_p) - Q^0(T_p) \vert \ \stas \ 0.$
\end{lemma}
The proof of this result can be found in Appendix B, and it is partially based on the pointwise properties of the integrands of $Q^0$ and $\wh Q^0_m$ shown in \cite{berrendero2018}.

For the case of greater order $FCAR(q)$ with $q>1$, this result, and therefore all the results of the next section, hold whenever the process $Z_{n,q}$ fulfils assumptions H1 and H2 of Lemma \ref{Lemma:UnifConv}. This is equivalent to suppose that all the products $X_i(t)X_j(s)$ satisfy H1 and H2 for $0\leq i,j \leq q-1$.


\section{Asymptotic results}\label{Sec:Asymp}

In the previous section we introduced the variable selection method in a general context. We have presented an optimality criterion that can be used to select the $p$ most relevant points, without imposing additional restrictions to the model that generates the data. If we assume that the data is generated by the finite dimensional model of Equation \eqref{Eq:SparseX}, additional asymptotic results for the estimator $\wh T_{p,m}$ and for the estimated curves can be derived. Therefore, in this section we assume that the kernel depends only on $p^*$ points:
\begin{equation}\label{Eq:Sparse}
    \phi(s,\cdot) \ = \ \sum_{j=1}^{p^*} \alpha_j(s) c_0(t_j^*,\cdot),
\end{equation}
for all $s\in[0,1]$, where $p^*$ is the minimum integer such that this expression holds. We denote by $T^* = T_{p^*}^* \in \Theta_{p^*}$ the set of points that generate the model. Two questions arise; how good is our estimator $\wh T_{p^*,m}$ when searching the real points $T^*$, and how can we estimate the real number of points $p^*$.


\subsection{Estimated points and trajectories}

From the expression of $Q$ given in Equation \eqref{Q1}, one sees that given \eqref{Eq:Sparse}, the set $T^*$ is a global minimum of $Q$ on $\Theta_{p^*}$, and therefore a global maximum of $Q^0$. Assuming for now that $p^*$ is known, we can prove that this optimum is unique and the estimated points ${\wh T}_{p^*,m}$ converge to the real ones $T^*$. This result is an extension of Theorem 1 of \cite{berrendero2018} and its proof is included in Appendix B.

\begin{theorem}\label{Theo:Ts}
Under the assumptions of Lemma \ref{Lemma:ContConv} for $p=p^*$, whenever \eqref{Eq:Sparse} holds and the covariance matrices $\Sigma_{T_{p^*}\cup S_{p^*}}$ are invertible for all $T_{p^*},S_{p^*}\in \Theta_{p^*}$ with $T_{p^*}\neq S_{p^*}$, then:
\begin{enumerate}
\item[(a)] The vector $T^*\in\Theta_{p^*}$ is the only global maximum of $Q^0$ on this space.
\item[(b)] $\widehat{T}_{p^*,m} \overset{a.s.}{\rightarrow} T^*$ with the sample size $m\to\infty$, where $\widehat{T}_{p^*,m}$ is given in Eq.\eqref{EstT} with $p=p^*$.
\item[(c)] $\widehat{T}_{p^*,m}$ converges to $T^*$ in quadratic mean when $m\to\infty$.
\end{enumerate}
\end{theorem}

Once that we have selected the most relevant points from the sample, we want to estimate the trajectories of the process. That is, we want to approximate 
\begin{equation}\label{Eq:Xreal}
    X_{n}(\cdot)_{T^*} \ = \ \alpha_1(\cdot)X_{n-1}(t_1^*) + \ldots + \alpha_{p^*}(\cdot)X_{n-1}(t_{p^*}^*).
\end{equation}
In the proof of Proposition \ref{Prop:Equiv} we have seen that the functions $({\alpha}_1(\cdot),\ldots,{\alpha}_{p^*}(\cdot))'$ used to carry out this projection are given by ${\Sigma}_{T_{p^*}}^{-1} ({c_1}(\cdot,t_1^*),\ldots,{c_1}(\cdot,t_{p^*}^*))$. Therefore, we can construct the estimated curves as
$\widehat{X}_{n}(\cdot)_{\widehat{T}_{p^*,m}}$ as $\widehat{\alpha}_1(\cdot)X_{n-1}(\widehat{t}_1) + \ldots + \widehat{\alpha}_{p^*}(\cdot)X_{n-1}(\widehat{t}_{p^*}),$
where now the functions $({\wh \alpha}_1(\cdot),\ldots,{\wh \alpha}_{p^*}(\cdot))'$ are computed using the sample version of the covariances as $\widehat{\Sigma}_{\wh T_{p^*},m}^{-1} (\widehat{c_1}(\cdot,\widehat{t}_1),\ldots,\widehat{c_1}(\cdot,\widehat{t}_{p^*}))$. Thus our proposed estimator for $X_{n}(\cdot)_{T^*}$ is
\begin{equation}\label{Eq:Resp}
\widehat{X}_{n}(\cdot)_{\widehat{T}_{p^*,m}} \ = \ {\wh c}_1(\cdot,\wh T_{p^*,m})' \, \widehat{\Sigma}_{\wh T_{p^*,m}}^{-1}X_{n-1}({\wh T_{p^*,m}}).
\end{equation}
Under the same conditions of the previous theorem, we can see that this estimator converges to $X_{n}(\cdot)_{T^*}$ uniformly a.s. and in quadratic mean.

\begin{theorem}\label{Theo:Resp}
Under the same assumptions of Theorem \ref{Theo:Ts} and when the sample size $m\to\infty$,
\begin{enumerate}
\item[(a)] $\widehat{X}_{n}(\cdot)_{\widehat{T}_{p^*,m}}$ converges to $X_{n}(\cdot)_{T^*}$ a.s. in $C[0,1]$: \ $\sup_s\big\vert \widehat{X}_{n}(s)_{\widehat{T}_{p^*,m}} - X_{n}(s)_{T^*}\big\vert \stas 0$.
\item[(b)] If, in addition, there exists $\eta>1$ such that $\E \Vert \vert X_n \vert^{2\eta} \Vert<\infty$, then one also gets $\E[(\sup_s \vert \widehat{X}_{n}(s)_{\widehat{T}_{p^*,m}} - X_{n}(s)_{T^*}\vert)^2] \to 0$. 
\end{enumerate}
\end{theorem}

The proof of this theorem is based on the one of Theorem 2 of \cite{berrendero2018} and can be found in Appendix B. In fact, as shown in Proposition 2 of \citet{mokhtari2003}, if the true kernel of the model is as in Equation \eqref{Eq:SparseGeneral}, the best linear predictor based on $X_{n-1}(t_1),\ldots,X_{n-1}(t_p)$ is the best probabilistic predictor of $(X_n-\varepsilon_n)$.


\subsection{Number of relevant points}

We are left with deriving an estimator for $p^*$, the number of points to select. Notice that the optimality criterion (Equation \eqref{Q1}) does not reach its minimum value for $p<p^*$. Additionally, there is no room for improvement using $p>p^*$.
Therefore, the number $p^*$ would be the smallest $p$ such that the minimum value of $Q(T_p)$ (or the maximum of $Q^0$) remains unchanged when increasing $p$.

The sample version of this idea would be as follows. Defining 
$$\Delta = \min_{p<p^*} (Q^0(T_{p+1}^*)-Q^0(T_p^*))>0$$ 
and fixing some $0 < \epsilon < \Delta$, we set
\begin{equation}\label{Eq:EstP}
\widehat{p}_m \ = \ \min\Big\{ \, p \, : \, \max_{T_{p+1}\in\Theta_{p+1}} \big\{\widehat{Q}_m^0(T_{p+1})\big\} - \max_{T_p\in\Theta_p} \big\{\widehat{Q}_m^0(T_p) \big\} < \epsilon \, \Big\}.
\end{equation}
The value $\Delta$ is merely theoretical, and should be understood as taking an $\epsilon$ small enough. The issue of choosing the threshold $\epsilon$ in practice is tackled in the following section.

This estimator is a.s. consistent for the real number of relevant variables.

\begin{theorem}
Suppose that assumptions of Lemma \ref{Lemma:ContConv} hold for $p\leq p^*$ and that $p^*$ is the smallest integer such that Equation \eqref{Eq:Sparse} is satisfied. Then the estimator given by Equation \eqref{Eq:EstP} fulfills $\widehat{p}_m\stas p^*$.
\end{theorem}
\begin{proof}
We can prove similar results as the ones given in Lemma 4 of \cite{berrendero2018} using the same reasoning, with the only difference that now $Q^0(T^*) = \int_0^1 \Vert \phi(s,\cdot) \Vert_{\calh}^2 \dd s$ (as in the proof of Proposition \ref{Prop:Equiv}).
\end{proof}


\section{Experimental setting}\label{Sec:ExpSetting}

In this section we introduce the data sets which appear along the experiments (both simulated and real), as well as other methods of the literature used for comparison. We start making a couple of theoretical comments that ease the implementation of the method. A general pseudo-code is provided in Appendix C.

\subsection{Practical considerations}

\textit{Greedy approximation}

In order to obtain the most relevant points in practice, we maximize the expression of $\wh Q^0_m$ (given in Equation \eqref{Eq:Qn}) in the $p$-dimensional space $\Theta_p$. However, due to computational limitations, this optimization is not feasible even for relatively small values of $p$. Therefore, a greedy approximation is carried out. We can decompose the function $Q^0$ in a way that directly suggests an iterative approximation to this optimization problem. If the vector $T_{p+1}\in\theta_{p+1}$ is such that it contains all the entries of $T_p$ plus a new one $t_{p+1}\in [0,1]$, using Equation (16) of \cite{berrendero2018} we can write,
\begin{eqnarray*}
Q^0(T_{p+1}) & = & \int_0^1 c_1(s,T_{p+1})' \ \Sigma_{T_{p+1}}^{-1} \ c_1(s,T_{p+1}) \ \dd s \\
& = & Q^0(T_p) \ + \  \frac{\int_0^1 \cov\big(X_n(s) - X_n(s)_{T_p}, X_{n-1}(t_{p+1})\big)^2 \ \dd s}{\var\big(X_{n-1}(t_{p+1}),X_{n-1}(t_{p+1})_{T_p}\big)},
\end{eqnarray*}
where $X_n(\cdot)_{T_p}$ and $X_{n-1}(\cdot)_{T_p}$ are the same kind of projections defined in Equation \eqref{Eq:Xreal}. For each fixed value $s\in[0,1]$, the integrand of this quotient recalls the classical multivariate ``forward selection'' method. The crucial difference here is that the variables are not selected among a fixed finite set but among the whole interval $[0,1]$. This derivation can be also done using the sample counterpart of $Q^0$,
$$\wh Q^0_m (T_{p+1}) \ = \ \wh Q^0_m(T_p) + \frac{\int_0^1 \wh\cov \big(X_n(s) - X_n(s)_{T_p}, X_{n-1}(t_{p+1})\big)^2 \ \dd s}{\wh\var\big(X_{n-1}(t_{p+1}),X_{n-1}(t_{p+1})_{T_p}\big)}.$$
Then the proposed algorithm selects at each step the point $t_{p+1}$ that maximizes this quotient. The starting point would be the one that maximizes $\wh Q^0_m(t) = \wh c_0(t,t)^{-1} \int_0^1 \wh c_1(s,t)^2 \dd s$. As usual when dealing with greedy algorithms, this approximation does not guaranty that the global maximum of $\wh Q^0_m$ is reached. However it performs well in practice. 

In order to compute this quotient we can use a similar reasoning as in the proof of Proposition 2 of \cite{berrendero2018} and rewrite the previous equation as,
\begin{equation}\label{Eq:QnIterM}
    \wh Q^0_m (T_{p+1}) \ = \ \wh Q^0_m(T_p) \ + \ \frac{\int_0^1  \big(\wh c_1(s,T_p)' \ \wh\Sigma_{T_p}^{-1} \ \wh c_0(t_{p+1},T_p) - \wh c_1(s,t_{p+1})\big)^2 \ \dd s}{\wh c_0(t_{p+1},t_{p+1}) - \wh c_0(t_{p+1},T_p)' \ \wh\Sigma_{T_p}^{-1} \ \wh c_0(t_{p+1},T_p)},
\end{equation}
where $\wh c_1(s,T_p)$ is the vector whose entries are given by the sample covariances $\wh\cov (X_n(s),$ $X_{n-1}(t_j))$, and equivalently for $\wh c_0$. 

\

\textit{Grid and covariance matrices}

Due to computational limitations, the search of the most relevant points should be done on a grid of $[0,1]$. If the data is given in a discretized fashion, then the grid is directly given by the data. However, if it is fully functional, the grid can be defined (theoretically) arbitrarily fine. Under the assumption that all the covariance matrices $\Sigma_{T_p}$ are invertible for $T_p$ in this grid, the quotient of Equation \eqref{Eq:QnIterM} is easy to compute. In addition, depending on the nature of the data, the estimations of the covariances $\wh c_1(\cdot,T_p)$ can be made fully functional or using the values on the fixed grid. Both possibilities are implemented for the numerical study.

However, for some real data sets the condition of the invertibility of $\Sigma_{T_p}$ may not be satisfied. For instance, if the curves are represented using a Fourier basis, this condition is rarely satisfied, since periodicity is introduced on the data. Thus, other representations should be used, like splines. If the data is in any case not invertible, it can be always preprocessed to remove the conflictive points of the grid. It would not affect the efficiency of the method, since these points would be linearly dependent of the others, so their information would be redundant. In addition, if the data set contains outliers, there exist several well-known methods to robustly estimate the covariance matrices. Furthermore, if we have some additional information about the covariance structure of the data, it can be directly incorporated to the method.

For the model of order greater than one, we substitute the sample lagged covariance function $\wh c_1(s,t)$ by the sample version of the picewise-defined function of Equation \eqref{Eq:covZq}.

\

\textit{Number of selected points}

In order to apply the estimator of Equation \eqref{Eq:EstP}, some value for $\epsilon$ must be fixed. In other words, it should be determined for which value of $p$ the quotient of Equation \eqref{Eq:QnIterM} has converged to zero. The standard approach to this problem is to fix the parameter by cross-validation. Additionally we test the proposal given in  \cite{berrendero2018}; to apply the usual k-means with $k=2$ to the logarithms of the values of the quotient in Equation \eqref{Eq:QnIterM}. If we denote by $L_m(p)$ the values of these logarithms, $\wh p_m$ would be the minimum $p$ such that all the $L_m(p)$ for $p>\wh p_m$ do not belong to the same cluster as $L_m(1)$. This is equivalent to fix the parameter $\epsilon$ of Equation \eqref{Eq:EstP} to $L_m(q)$ where $q$ is the largest value such that $L_m(q)$ belongs to the same cluster as $L_m(1)$.


\subsection{Methodology}

We compare the efficiency of the proposal with two other recent methods. Both of them carry out the dimension reduction using functional principal components. For the forecasting experiments we also compare with two ``base'' methods that do not reduce dimension. We indicate in brackets the names used in the tables for each method, which are included in Appendix C.

\begin{itemize}
    \item The method proposed in this paper (RKHS) has been implemented in four different ways. As mentioned in the previous section, we use two approaches to select the number of relevant variables; doing clustering on the maximum values of the $\wh Q^0_m$ functions (CL) and by cross-validation (CV). In addition, the points can be selected by using covariance vectors on a grid or computing purely functional lagged-covariance functions. We use one or the other depending on the nature of the data. 
    \item The method proposed in \citet{aue} (fFPE). This proposal uses a dimension reduction method based on functional principal components analysis to find a finite dimensional space on which the prediction is performed using a vector autoregressive model. The model order and dimension of the finite dimensional space are chosen by the fFPE criterion. For details, see \citet{aue}, where the empirical properties of the approach are demonstrated in depth. The \textsf{R}-code of this method was provided by the authors.
    \item The method proposed in \citet{bosq} and \citet{kokoreim} (KR). This prediction method by Bosq is the one known as the standard prediction method for functional autoregressive processes. To determine the order of the functional autoregressive model to be fitted, we use the multiple testing procedure of \citet{kokoreim}.
    \item Exact and Naive methods are implemented in order to provide some bounds on the errors. These methods are also used, for instance, in \citet{horvath2012}. The exact prediction consists in ``predicting'' the response directly as $\rho (x_{n-1})$. Therefore, it can be only applied for simulated data, since the operator $\rho$ is unknown for real data sets. It is not really a prediction method but gives us an idea of the minimum error that we can achieve. The Naive approach simply predicts $\wh x_n$ as $x_{n-1}$.
\end{itemize}

Both the maximum number of points to select and the number of principal components are always limited to 10. For the simulated data all methods are tested using a sample of size $n=115$, where 100 realizations are used for training and the remaining 15 for test. Each experiment has been replicated 100 times. For the real data sets we use a window moving approach with five blocks to obtain several measures of the errors. The size of the windows is adjusted depending on the sample size of each set. The order of the process is always limited to 3 for the methods. However, for our proposal we have to set it to order 1 whenever the curves can not be interpreted as $x_n(s)=z(s+n)$, with $z$ continuous.

Usually the functional data sets are given in a discretized fashion. Some of the tested methods require to transform previously the data to truly functional. However, our discrete proposal can also deal with discretized data. In addition, when the data is irregular, some information could be lost when transforming the data to functional. This complicates the comparison between the different methods. Therefore, for this kind of discretized data sets we measure two different types of errors.
\begin{itemize}
    \item Discrete errors: The error is measured using the original discretized data. The discrete version of the proposal (the one that uses covariance vectors) is tested. The predictions returned by the methods that use fully functional data are evaluated on the same grid given by the data.
    \item Functional errors: The data are transformed to functions using a Bsplines basis before applying the methods. We have found that using Bsplines is more suitable in this setting, since the standard Fourier basis introduces periodicity in the data. For the displayed results 10 functions of the basis are used, but different numbers have been tested without significant changes. The functional implementations of the RKHS proposal, which estimate the whole lagged-covariance functions in a purely functional way, are tested now.
\end{itemize}
As we see in the following subsection, one of the simulated sets is purely functional. In this case only the functional errors are measured. Two different norms are used to measure the error: the standard $L^2[0,1]$ norm and the supremum norm of $C[0,1]$ that has been used along the paper. Each of these norms measure different characteristics of the predictions. We also measure two different relative errors, given a sample of curves $x_1,\ldots,x_m$,
\begin{equation}\label{Eq:Error}
e_1 = \sum_{i=1}^m \frac{\Vert x_i - \wh x_i \Vert}{\Vert x_i \Vert}, \hspace*{1cm} e_2 = \frac{\sum_{i=1}^m \Vert x_i - \wh x_i \Vert}{\sum_{i=1}^m\Vert x_i \Vert}.
\end{equation}
The first one gives the same importance to all curves regardless of their norm, while the second one place more importance to the errors in the curves of biggest norms, since it is just a scaling of the absolute error.


\subsection{Simulated data}

We test the different methods using simulated sets that fulfil the sparsity assumption of Equation \eqref{Eq:Sparse} as well as some which not. Most of them are inspired by other data sets used in the literature. Some realizations of these processes can be found in Figure \ref{Fig:TraySim}.

\begin{itemize}
    \item Two data sets satisfying the sparsity assumption with standard Brownian innovations. The real points are $T^*=(0.3,0.5,0.9)$ with two different sets of functions $\alpha_j$. The first ones are logarithms, $\log((1+s)j^{-1})$ for $j=1,2,3$ and $s\in[0,1]$, similar to the function used for the simulated data in \citet{berrendero2018}. The second set of functions is $sin(30\pi j^{-1}s)$, since we also want to incorporate a data set with high variation. When transforming this last data set to purely functional, 30 Bspline functions are used instead of 10, to be able to capture most of the variation.
    \item Ornstein-Uhlenbeck process introduced in Example \ref{Ex:OU}. This is the only simulated set for which $X_n(s)=Z(s+n)$, so that we can use the model $FCAR(3)$. 
    \item FAR process with linearly decaying eigenvalues of the covariance operator ($s=1,\ldots,15$). Following the simulation example used in \citet{aue}, this set consists of spanning a $D$-dimensional space by the first $D$ Fourier basis functions, and then generate random $D\times D$ parameter matrices and a $D$-dimensional noise process, where the construction ensures a linear decay of the eigenvalues of the covariance operator. The slow decay of these eigenvalues makes sure that problems with PCA based methods due to non-invertibility of the covariance operator are avoided. In this example only the functional errors are measured since it is purely functional by construction. 
\end{itemize}

\begin{figure}
\centering
\begin{subfigure}[b]{0.24\textwidth}
	\centering
    \includegraphics[scale=0.22]{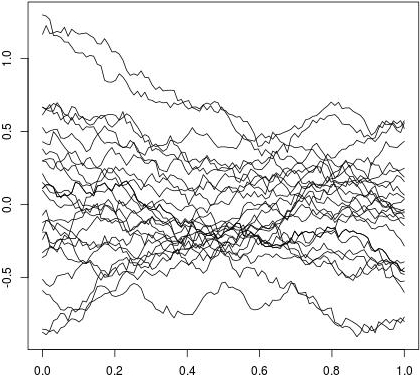}
    \caption{\footnotesize Sparse with logs.}
\end{subfigure}
\begin{subfigure}[b]{0.24\textwidth}
	\centering
    \includegraphics[scale=0.22]{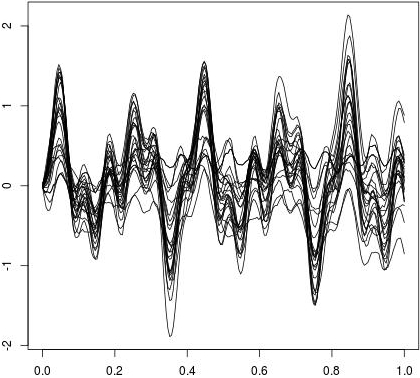}
    \caption{\footnotesize Sparse with sins.}
\end{subfigure}
\begin{subfigure}[b]{0.24\textwidth}
	\centering
    \includegraphics[scale=0.22]{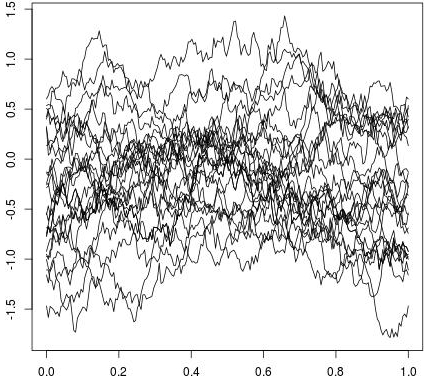}
    \caption{\footnotesize O.U.}
\end{subfigure}
\begin{subfigure}[b]{0.24\textwidth}
	\centering
    \includegraphics[scale=0.22]{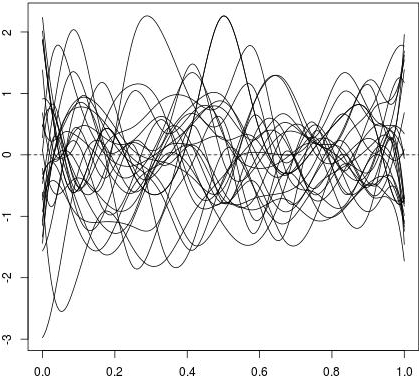}
    \caption{\footnotesize FAR}
\end{subfigure}
\caption{\footnotesize 25 trajectories of each of the simulated data sets.}
\label{Fig:TraySim}
\end{figure}


\subsection{Real data sets}

We test also some real data sets, a couple of them already used in other recent papers.

\begin{itemize}
    \item Particulate matter concentrations (PM10). This data set is used, for instance, in \citet{aue} and consists on 175 samples. It contains the $\mu gm^{-1}$ concentration in air of a particular substance with aerodynamic diameter less than 10 $\mu m$. The measures were taken each half hour from from October 1, 2010 to March 31, 2011 in Austria. The data is preprocessed in the same way as suggested in \cite{aue}. For the five windows we take blocks of 115 observations, 100 for training and 15 for test.
    \item Vehicle traffic data (Traffic) presented in \citet{auekle}. The original data set was provided by the Autobahndirektion S\"udbayern. It contains the amount of vehicles traveling each five minutes on the highway A92 in Southern Bavaria, Germany, from January 1 to June 30, 2014. Retaining only working days, we work with 119 samples divided into 5 windows of size 99; 94 for train and 5 for test. 
    \item Indoor temperature of a ``solar house'' (Temp). This data set consist in temperature measures each $15$ minutes during $42$ days in the living room of a SMLsystem solar house. The whole data set (which contains other different attributes) is studied in \cite{zamora2014} and it is available in \url{http://archive.ics.uci.edu/ml/datasets/SML2010}. This is the smallest set, so it is divided into 5 windows of size 34, from which just 2 curves are used for test. For this set we were forced to use at most 9 PCA components for the fFPE method, in order to avoid computational errors.
    \item Utility demand data (Utility) which appears in the book \cite{hyndman2008} and is available in the \textsf{R} package ``expsmooth''. The original set is made of 126 curves of hourly utility demand from a company of the the Midwestern United States, starting on January 2003. Since this work is focused on variable selection for data sampled on a fine grid, the curves have been sub-sampled to simulate observations each 15 minutes. The five windows into which the curves are splited consist in 100 samples for train and 5 for test.
\end{itemize}

A few curves of these data sets are included in Figure \ref{Fig:TrayReal}. 

\begin{figure}
\centering
\begin{subfigure}[b]{0.24\textwidth}
	\centering
    \includegraphics[scale=0.22]{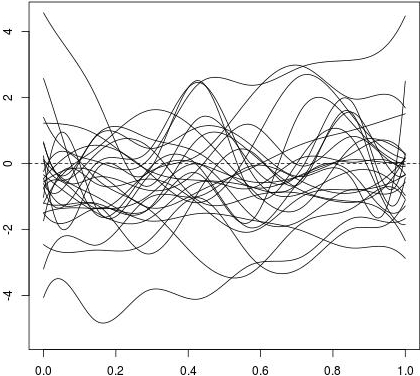}
    \caption{\footnotesize Functional PM10.}
\end{subfigure}
\begin{subfigure}[b]{0.24\textwidth}
	\centering
    \includegraphics[scale=0.22]{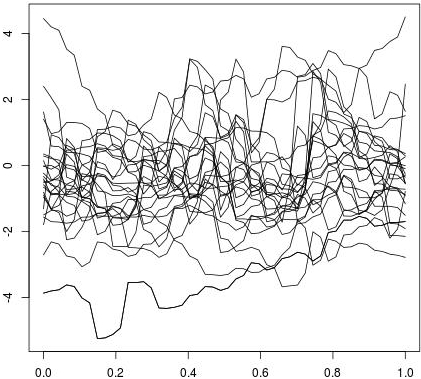}
    \caption{\footnotesize Original PM10.}
\end{subfigure}
\begin{subfigure}[b]{0.24\textwidth}
	\centering
    \includegraphics[scale=0.22]{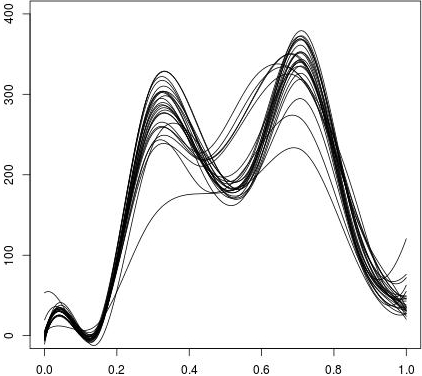}
    \caption{\footnotesize Functional traffic.}
\end{subfigure}
\begin{subfigure}[b]{0.24\textwidth}
	\centering
    \includegraphics[scale=0.22]{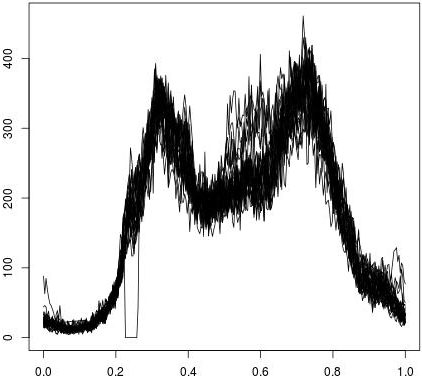}
    \caption{\footnotesize Original traffic.}
\end{subfigure}
\begin{subfigure}[b]{0.24\textwidth}
	\centering
    \includegraphics[scale=0.22]{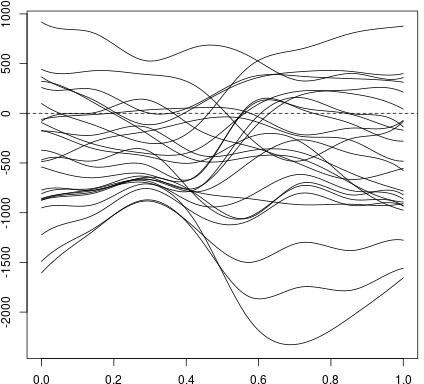}
    \caption{\footnotesize Functional temp.}
\end{subfigure}
\begin{subfigure}[b]{0.24\textwidth}
	\centering
    \includegraphics[scale=0.22]{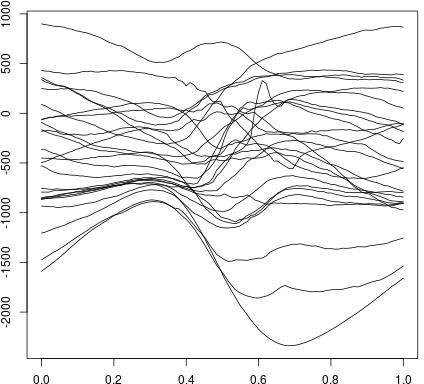}
    \caption{\footnotesize Original temperature.}
\end{subfigure}
\begin{subfigure}[b]{0.24\textwidth}
	\centering
    \includegraphics[scale=0.22]{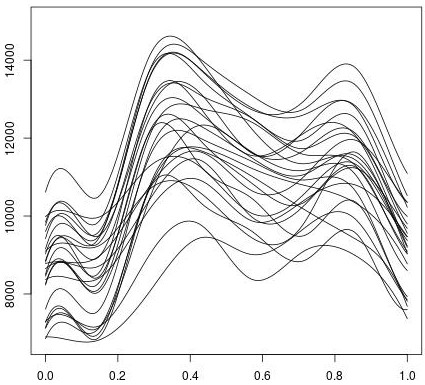}
    \caption{\footnotesize Functional utility.}
\end{subfigure}
\begin{subfigure}[b]{0.24\textwidth}
	\centering
    \includegraphics[scale=0.22]{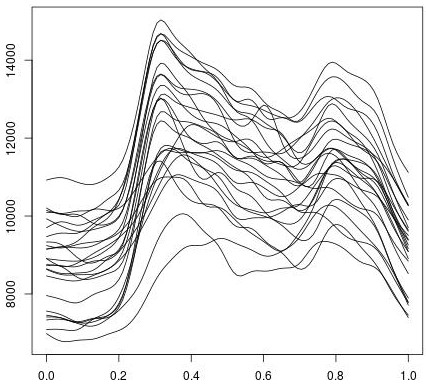}
    \caption{\footnotesize Original utility.}
\end{subfigure}
\caption{\footnotesize 25 trajectories of the real data sets, both discrete and functional.}
\label{Fig:TrayReal}
\end{figure}


\section{Experiments and results}\label{Sec:Experiments}

In this section we present the results of the experiments, conducted in order to check the performance of the proposal. The code used to run these simulations is provided and the tables containing the results can be found in Appendix C. The entries marked with bold letters there correspond to the best performance in each case. We want to emphasize that we do not intend to obtain definitive and general conclusions from these results. We believe that only real applications of the techniques lead to safe conclusions.

\subsection{Forecasting}

The main goal for which our proposal is designed is the prediction of time series. Accordingly, the greatest part of the experiments is devoted to forecasting. 

\

\textit{Simulated data sets}

Table \ref{Table:sim} of Appendix C summarizes the measurements for the simulated data sets of the two types of errors $e_1$ and $e_2$ (Equation \eqref{Eq:Error}). Regarding our two proposals, there is not a method that uniformly outperform the other one. That is, both cluster and cross-validation perform well when it comes to select the number of points. In general, our proposals are mainly the victors, closely followed by the FPCA approach with the fFPE criterion. These are the expected results, since three out of the four data sets fulfill the sparse model of Equation \eqref{Eq:SparseX}. In any case, our proposal also slightly outperforms the others for the FAR data, where this sparsity assumption is far from being satisfied.

\

\textit{Real data sets}

In Table \ref{Table:real} of Appendix C we summarize the different error measurements for the four real data sets tested. Taking these results into account, it is even less clear which implementation of our proposal, the cross-validation one or the cluster one, is the best choice. For the two first data sets it seems that the FPCA approach with fFPE slightly outperforms the other methods. However, the differences between it and our proposals are in general small, even achieving the same error, or improving it, in about half of the measures. By comparison, our proposal is the victor for the last two data sets. It is particularly noteworthy the differences obtained for the temperature data set, which is the smallest one with only $32$ curves for training in each window. The error measurements of our proposals for this set fall in the interval $[0.24,0.82]$, while the measurements for fFPE are in $[1.95, 5.3]$. This could be due to the simplicity of our proposal, which just relies on the computation of the covariance matrix of (at most) 10 real random variables. This simplicity is also reflected in the execution time presented later.

In addition, for these real data sets we have also obtained the selected points, which are shown in Figure \ref{Fig:SelPoints} (these curves are centered versions of the ones in Figure \ref{Fig:TrayReal}). It is difficult to reach meaningful conclusions for the four sets altogether, but we can make a couple of interesting observations. For instance, the points selected for the discrete data sets are more ``precise'' (in some sense) than the ones for the functional version, which look more equispaced. This could lead to think that we are ``dispersing'' the dependence of the data when representing them on a functional basis. In addition, the points selected with cluster and cross-validation for the discrete sets are similar, although it seems that the cross-validation implementation selects more points than needed (in view of the prediction performance). 

We can also analyze the points selected for each data set separately. We mainly focus on the points selected for the discrete versions of the data. For the pollution data set, it seems that the last few hours of the day are the most informative when predicting the pollution of the following day, which seems reasonable. But it seems also important to measure the pollution early in the morning (since all the methods select at least one point in the interval $[0.2,0.4]$, which would correspond to between 5:00 and 9:00). With regard to the traffic data, we can identify the most relevant time interval around 17:00, which coincide with one of the moments of greatest traffic volume. All the methods select one point around 13:00 as well, which correspond to the local minimum in the original curves. For the temperature it looks like almost the only relevant hours for prediction purposes are between 0:00 and 2:00 of the previous day (since the blue points correspond to $X_{n-2}$), along with around 8:00 in the morning. We find this result remarkable, since it is not completely intuitive. Finally, for the utility data set the most noteworthy fact is that the cluster implementation for the discrete version does not select any point for $X_{n-1}$, which would mean that the dependence lies more backward in time.

\begin{figure}
\centering
\begin{subfigure}[b]{0.95\textwidth}
	\centering
    \includegraphics[scale=0.6]{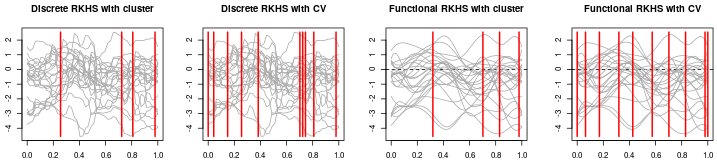}
    \caption{\footnotesize Selected points for PM10.}
\end{subfigure}
\begin{subfigure}[b]{0.95\textwidth}
	\centering
	\vspace*{5mm}
    \includegraphics[scale=0.6]{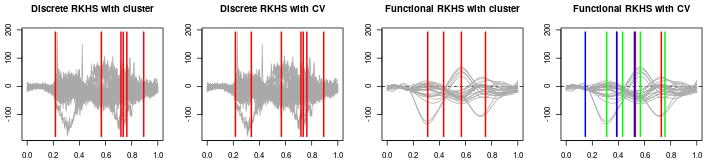}
    \caption{\footnotesize Selected points for traffic.}
\end{subfigure}
\begin{subfigure}[b]{0.95\textwidth}
	\centering
	\vspace*{5mm}
    \includegraphics[scale=0.6]{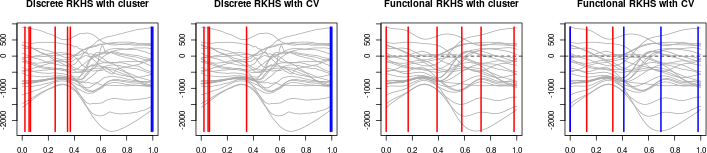}
    \caption{\footnotesize Selected points for temperature.}
\end{subfigure}
\begin{subfigure}[b]{0.95\textwidth}
	\centering
	\vspace*{3mm}
    \includegraphics[scale=0.6]{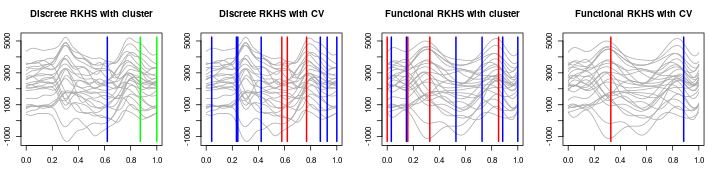}
    \caption{\footnotesize Selected points for utility.}
\end{subfigure}
\caption{\footnotesize Selected points in $X_{n-1}$ (solid), $X_{n-2}$ (dashed) and $X_{n-3}$ (dotted).}
\label{Fig:SelPoints}
\end{figure}


\subsection{Execution time}

We measured the execution times of all the previous forecasting experiments. Both the functional (funct) and the discrete (disc) implementations of our proposal are measured. Table \ref{Table:timesReal} of Appendix C shows the mean execution times of each of the methods for the real data sets. It seems that working with the transformed functional data is slower in general, and that our two discrete implementations are considerably faster than the other methods. The traffic data set is the only one for which our proposal is not the fastest one. This is due to the larger size of the grid, since the curves are sampled every five minutes. Since our procedure checks almost all the points of the grid at each step, the grid size notably affects the execution time.

We also measure how the sample size affects the execution time, increasing it from 50 to 250 samples for the four simulated data sets. The obtained results are available in Table \ref{Table:timesSim} of Appendix C and they are also plotted in Figure \ref{Fig:times}. We see that our two discrete implementations are almost not affected by the change on the sample size, compared with the other methods. The execution times for the functional implementations are also almost constant with the sample size, although we can see that for the O.U. the execution times are quite high. This is due to the use of the model $FCAR(3)$ for this data set instead of $FCAR(1)$. Therefore, we analyze also the impact of the order of the model on the execution time. We use the values $q=1,\ldots,5$ for this same data set. The results are summarized in Table \ref{Table:timesFCAR} of Appendix C. We can see that the value of this parameter significantly affects the execution time of the functional implementations.

\begin{figure}
\centering
\begin{subfigure}[b]{0.45\textwidth}
	\centering
    \includegraphics[scale=0.4]{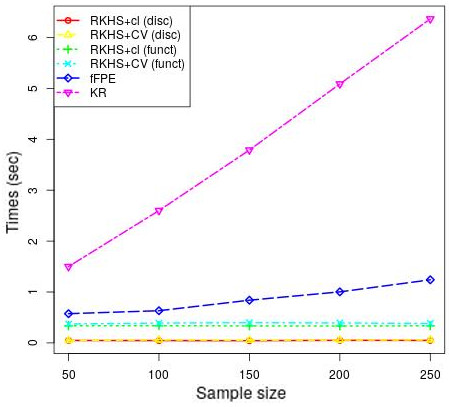}
    \caption{\footnotesize Sparse with logs.}
\end{subfigure}
\begin{subfigure}[b]{0.45\textwidth}
	\centering
    \includegraphics[scale=0.4]{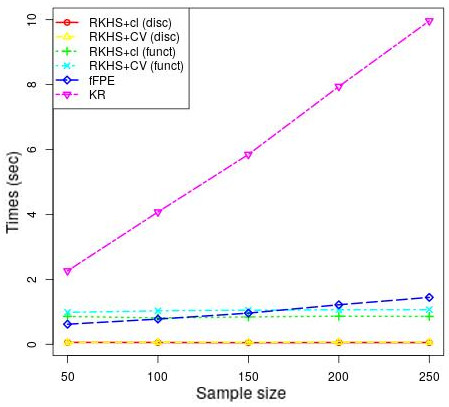}
    \caption{\footnotesize Sparse with sins.}
\end{subfigure}
\begin{subfigure}[b]{0.45\textwidth}
	\centering
    \includegraphics[scale=0.4]{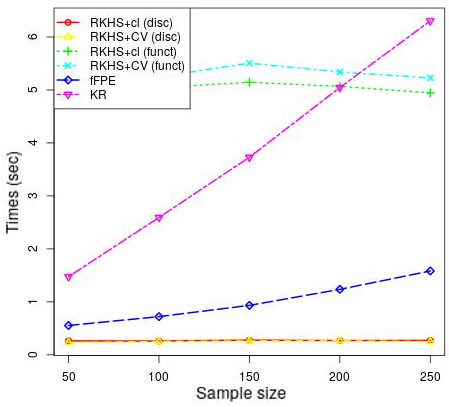}
    \caption{\footnotesize O.U.}
\end{subfigure}
\begin{subfigure}[b]{0.45\textwidth}
	\centering
    \includegraphics[scale=0.4]{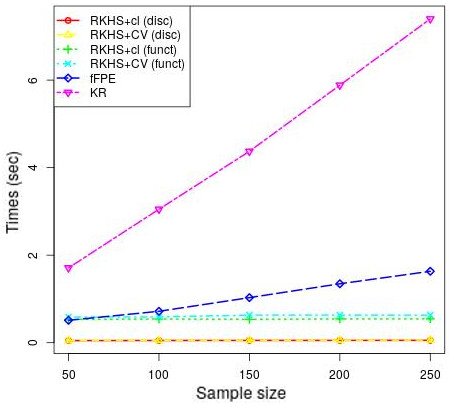}
    \caption{\footnotesize FAR.}
\end{subfigure}
\caption{\footnotesize Execution times for the simulated data sets when increasing the sample size.}
\label{Fig:times}
\end{figure}


\subsection{Kernel approximation}

As mentioned when the general model was first introduced in Equation \eqref{fullmodel}, any function $\phi(s,\cdot)\in\calh(X)$ can be arbitrarily well approximated by finite linear combinations as in Equation \eqref{Eq:SparseGeneral}, which we denote here as $\phi_p(s,\cdot)$. We quantify how good this approximation is when increasing the number of elements $p$ in this sum.

In order to avoid the use of samples, which introduces noise to the measurements, we work with the RKHS associated to the standard Brownian Motion, since in this case we know explicitly the space $\calh(X)$,
$$\calh(X) \ = \ \{f\in L^2[0,1] \ : \ f(0)=0, \ f \text{ absolutely continuous and } f'\in L^2[0,1]\},$$
where $f'$ denotes the derivative. The inner product of this space is given by
$$\langle f , g\rangle_\calh = \int_0^1 f'(s)g'(s) \dd s, \  \text{ for } f,g\in \calh(X).$$
Since the different norms in function spaces are not equivalent, it is not obvious which norm should we use to measure the errors. We decided on the $\calh(X)$-norm, since it is the one that appears in the theoretical results. In addition, the $L^2$-norm is always less than the RKHS-norm. Thus, the differences shown here are greater than the differences in $L^2[0,1]$.

We approximate different kernels in $\calh(X)$, increasing $p$ from $1$ to $3$. First, we use the sparse kernel function with logarithms previously used for the forecasting experiments. Then we try the functions
\begin{equation}\label{Eq:phisApprox}
\phi^1(s,t) = \cos (2\pi s) \sin(2\pi t), \hspace*{5mm} \phi^2(s,t) = \sin (2\pi s t) \ \text{ and } \ \phi^3(s,t) = -\log (5st+1).
\end{equation}
Different continuous functions can be tried using the provided code, as long as they fulfill $\phi(s,0)=0$ for all $s\in[0,1]$ and the derivatives of $\phi(s,\cdot)$ lie in $L^2[0,1]$.

Figure \ref{Fig:KernelApp} shows the approximated kernels. In Figure \ref{Fig:DistRKHS} we plot the distances $\Vert \phi(s,\cdot) - \phi_p(s,\cdot)\Vert_\calh$ for $s\in[0,1]$, for $p$ from $1$ to $20$. We can see that these distances go to zero for every $s$, although this convergence does not seem to be at the same rate for every point. As expected, the distances for the sparse representation are all zero for $p=3$ (the real model). 

\begin{figure}
\centering
\begin{subfigure}[b]{0.95\textwidth}
	\centering
    \includegraphics[scale=0.3]{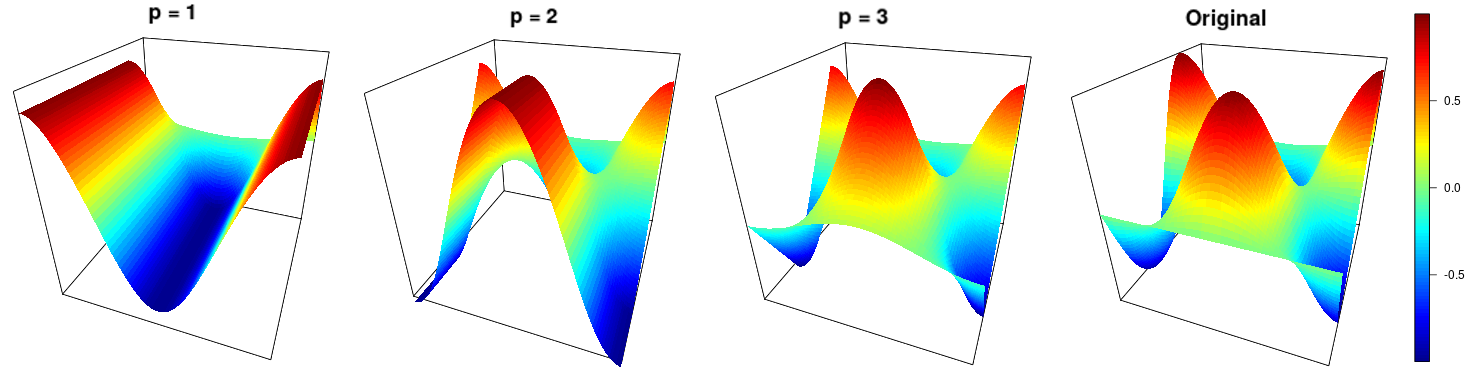}
    \caption{\footnotesize Function $\phi^1$ of Equation \eqref{Eq:phisApprox}.}
\end{subfigure}
\begin{subfigure}[b]{0.95\textwidth}
	\centering
    \includegraphics[scale=0.3]{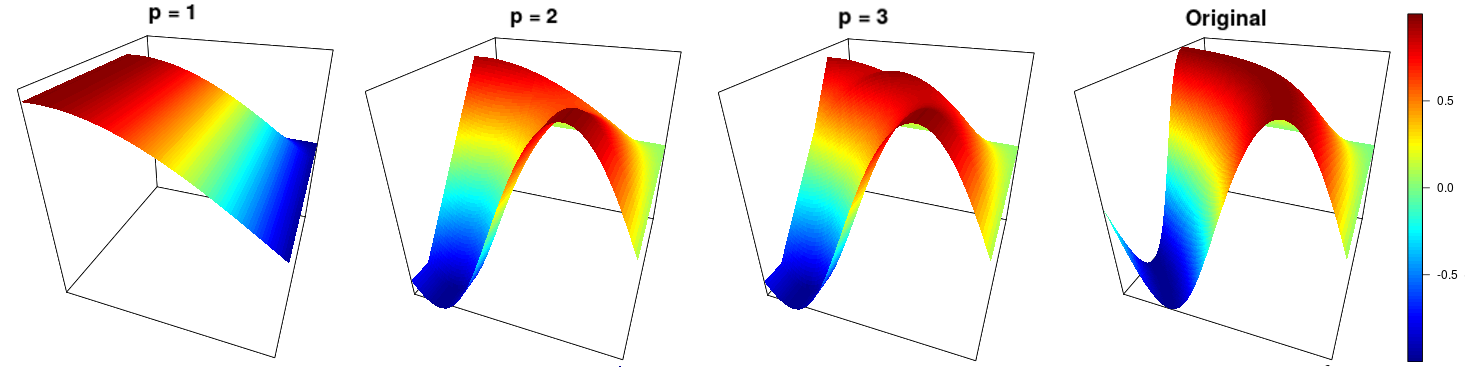}
    \caption{\footnotesize Function $\phi^2$ of Equation \eqref{Eq:phisApprox}.}
\end{subfigure}
\begin{subfigure}[b]{0.95\textwidth}
	\centering
    \includegraphics[scale=0.3]{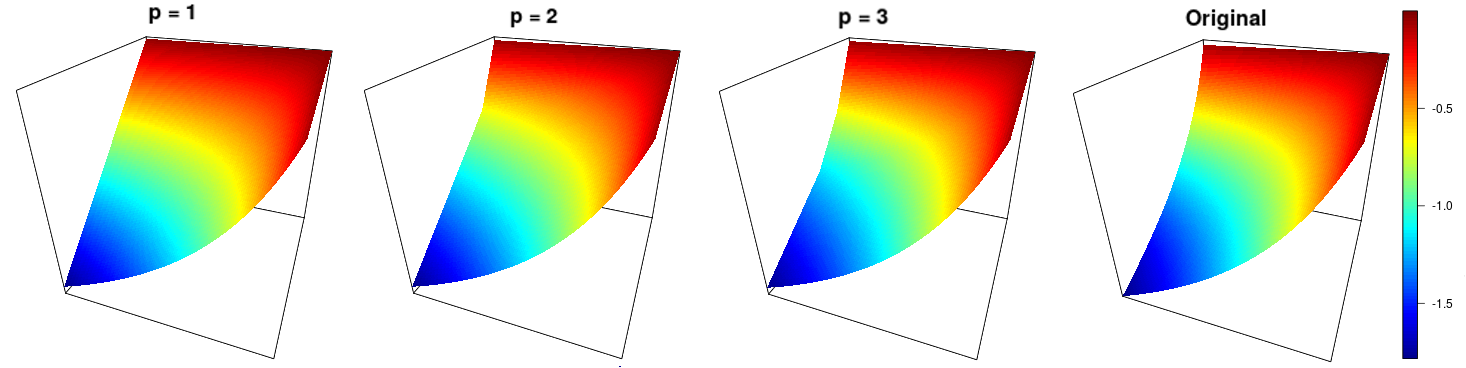}
    \caption{\footnotesize Function $\phi^3$ of Equation \eqref{Eq:phisApprox}.}
\end{subfigure}
\begin{subfigure}[b]{0.95\textwidth}
	\centering
    \includegraphics[scale=0.3]{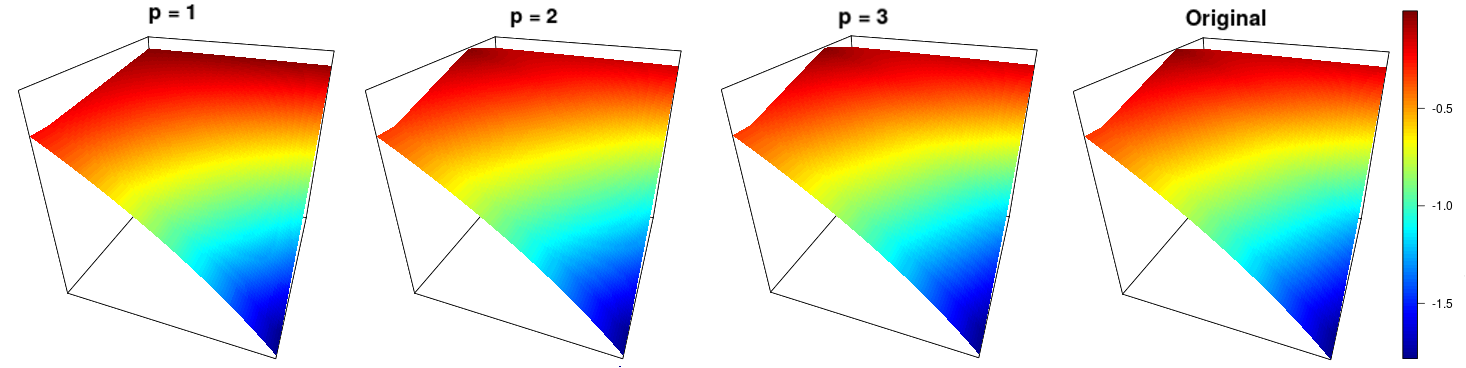}
    \caption{\footnotesize Sparse kernel with logarithms.}
\end{subfigure}
\caption{\footnotesize Approximations of the functions $\phi(s,\cdot)\in \calh(X)$ when increasing $p$ in Equation \eqref{Eq:SparseGeneral}.}
\label{Fig:KernelApp}
\end{figure}

\begin{figure}
\centering
\begin{subfigure}[b]{0.24\textwidth}
	\centering
    \includegraphics[scale=0.22]{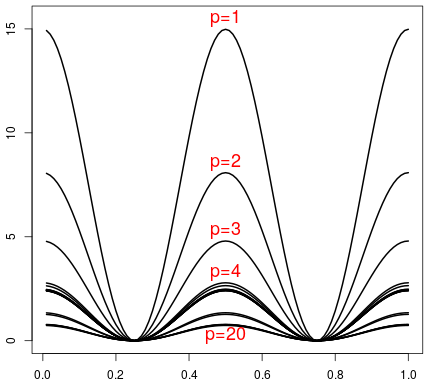}
    \caption{\footnotesize $\phi^1$ of Eq. \eqref{Eq:phisApprox}.}
\end{subfigure}
\begin{subfigure}[b]{0.24\textwidth}
	\centering
    \includegraphics[scale=0.22]{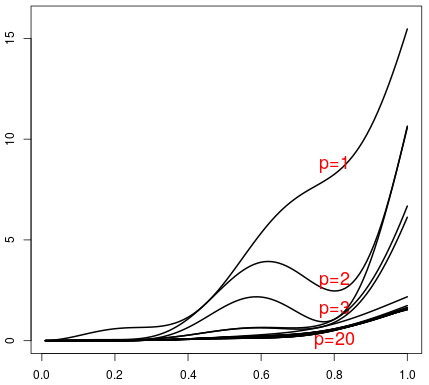}
    \caption{\footnotesize $\phi^2$ of Eq. \eqref{Eq:phisApprox}.}
\end{subfigure}
\begin{subfigure}[b]{0.24\textwidth}
	\centering
    \includegraphics[scale=0.22]{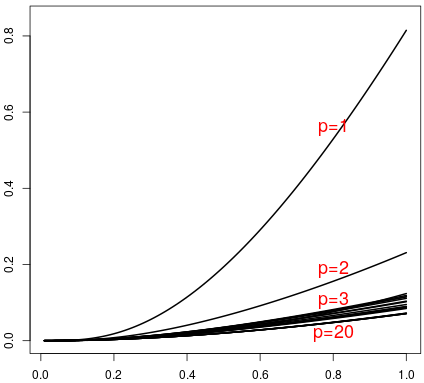}
    \caption{\footnotesize $\phi^3$ of Eq. \eqref{Eq:phisApprox}.}
\end{subfigure}
\begin{subfigure}[b]{0.24\textwidth}
	\centering
    \includegraphics[scale=0.22]{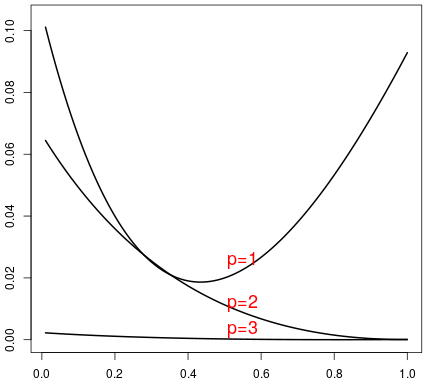}
    \caption{\footnotesize Sparse with logs.}
\end{subfigure}
\caption{\footnotesize Distances $\Vert \phi(s,\cdot) - \phi_p(s,\cdot)\Vert_\calh$ for $s\in[0,1]$.}
\label{Fig:DistRKHS}
\end{figure}


\section{Conclusions}

In the present paper we have fundamentally extended the theory developed for regression with scalar response and independent data, introduced by \citet{berrendero2018}, to the setting of prediction of functional time series, whose dependence is modeled using an autoregressive structure. That is, a variable selection technique for prediction is developed, based on the theory of Reproducing Kernel Hilbert Spaces. This variable selection approach helps to overcome some of the usual problems coming from the use of other dimension reduction techniques. Besides, the change of environment from the standard $L^2[0,1]$ to $C[0,1]$ allows us to prove uniform almost sure convergence of the estimations. We also provide an a.s. consistent estimator for the number of relevant variables involved in the model, which is computationally efficient.

When compared with other prediction methods of the literature, our proposal is quite competitive. The results obtained for the real data sets tested are  encouraging. In addition, the execution times of our implementation are smaller than the competitors. That is, our proposals, particularly the discrete approaches, are more suitable for large data sets. Furthermore, the proposed estimators can be directly adapted to discrete or fully functional data sets.


\section*{Acknowledgements}
We are very grateful to Prof. Dr. Claudia Kl\"uppelberg, Antonio Cuevas, Jos\'e Ram\'on Berrendero and Alexander Aue and two referees for their help, suggestions and criticisms. We furthermore thank Siegfried H\"ormann and the Autobahndirektion S\"uedbayern for providing the datasets. This work has been partially supported by Spanish Grant MTM2016-78751-P and the European Social Fund (``Ayudas para contratos predoctorales para la formaci\'on de doctores 2015'' and ``Ayudas a la movilidad predoctoral para la realizaci\'on de estancias breves en centros de I+D 2016'', BES-2014-070460). The paper was partly written during the visit of Beatriz Bueno-Larraz as an academic guest researcher at the Technical University of Munich.

\bibliographystyle{abbrvnat}
\bibliography{bibfile}


\newpage
\phantomsection
\addcontentsline{toc}{section}{Appendix A}
{\large \textbf{APPENDIX A: Theoretical model justification}}

\

In this appendix we include the discussion about model of Equation \eqref{fullmodel} of the main document, 
$$X_n(\cdot) \ = \ \Psi_{X_{n-1}}^{-1}\big(\phi(\cdot,\star)\big) + \varepsilon_n(\cdot),  \qquad n\in\Z.$$
We work with $X_n, n\in\Z$ a centered stationary process taking values in $C[0,1]$ such that $\E\big[\big(\sup_s |X_n(s)|\big)^2\big]<\infty$. Note that, since the process is stationary, its covariance structure remains invariant and then the space $\calh(X)$, on which the Lo\`eve's isometry is defined, does not depend on $n$. In addition, at the end of this document we include a result concerning the convergence of the projections when increasing the number of selected points.

We start analyzing the pointwise definition of the previous model, for each $s\in[0,1]$, 
\begin{equation}\label{Eq:AppXs}
X_n(s) \ = \ \Psi_{X_{n-1}}^{-1}\big(\phi(s,\cdot)\big) + \varepsilon_n(s),  \qquad n\in\Z.
\end{equation}
This definition can be applied for any process $X_n$, $n\in\Z$, with $\E\big[\big(\sup_s |X_n(s)|\big)^2\big]<\infty$, since then $X_n(s)$ has finite variance and the Lo\`eve's isometry can be applied. Moreover,
\begin{equation}\label{Eq:phiC1}
\phi(s,\cdot) \ = \ \Psi_{X_{n-1}}\left( X_n(s) - \varepsilon_n(s) \right) \ = \ \E \left[ \left( X_n(s) - \varepsilon_n(s) \right) X_{n-1}(\cdot) \right] \ = \ c_1(s,\cdot),
\end{equation}
and then 
$$\Vert c_1(s,\cdot)\Vert_{\calh} = \Vert \phi(s,\cdot) \Vert_{\calh} = \var\big( X_n(s) - \varepsilon_n(s) \big)<\infty.$$
That is, the pointwise evaluations $X_n(s)$ can be written as $\Psi_{X_{n-1}}^{-1} (c_1(s,\cdot))  + \varepsilon_n(s)$, which is always well-defined. From Equation \eqref{Eq:phiC1} we see that, when changing the working space from $L^2[0,1]$ to $\calh(X)$, the solution of the model does not require to invert the covariance operator.

As pointed out in the main document, the finite FCAR model depending on $p$ points is recovered when $\phi$ is as in Equation \eqref{Eq:SparseGeneral}. We focus now on this finite family, since we are mainly interested in variable selection. In order to make sense of the functional definition of FCAR models (in the same vein as the general AR model \eqref{ARB} of the main document) and to be able to obtain some properties about the process, we make use of a more general family. In view of Definition \ref{def:fcar1}, one can understand that each realization $x_n$ equals $\rho(x_{n-1})+\varepsilon_n$, where $\rho$ is the operator: 
\begin{equation}\label{Eq:rho}
\rho(f) = \sum_{j=1}^p \alpha_j(\cdot) f(t_j) \hspace*{5mm} \text{for} \ f\in C[0,1].
\end{equation}
We prove with Proposition \ref{Prop:ARBapp} that this interpretation is well founded. Note that this operator depends on the covariance function $c_1(s,t)$, since it uses the same set of points $t_j$ and functions $\alpha_j$ that define it (by Equation \eqref{Eq:phiC1}). 

\begin{proposition}\label{Prop:ARBapp}
Let $X_n$ follow a FCAR(1) model of Defintion~\ref{def:fcar1} such that $\sum_{j=1}^p\Vert \alpha_j \Vert < 1$, then it has a unique strictly stationary solution given by 
\begin{align*}
x_n = \sum_{j=0}^\infty \rho^j (\varepsilon_{n-j}), \quad n\in\Z,
\end{align*}
where $\rho$ is defined in Equation \eqref{Eq:rho}, the series converges almost surely and besides $$\E\Big[\Big(\sup_s\Big\vert X_n(s) - \sum\limits_{j=0}^p \rho^j (\varepsilon_{n-j})(s)\Big\vert\Big)^2\Big]\to 0 \hspace*{3mm} \text{ as } p\to\infty.$$
\end{proposition}

\begin{proof}
This proof relies on the theory of Banach space valued autoregressive (ARB) processes (introduced in Definition~6.1 of \citet{bosq}) with $B=C[0,1]$. First we check that our model follows an ARB model as in \eqref{ARB} of the main document. That is, that the operator of Equation \eqref{Eq:rho} is bounded. It follows from the definition of the norm in the space of linear operators,
\begin{eqnarray*}
\Vert\rho\Vert_{\mathcal{L}} &=& \sup_{\Vert f \Vert \leq 1} \Big\Vert \sum_{j=1}^p \alpha_j(\cdot) f(t_j) \Big\Vert \ = \ \sup_{\Vert f \Vert \leq 1} \sup_{s\in [0,1]} \Big\vert \sum_{j=1}^p \alpha_j(s) f(t_j) \Big\vert \\
&\leq& \sup_{\Vert f \Vert \leq 1} \sup_{s\in [0,1]} \sum_{j=1}^p |\alpha_j(s)| \, | f(t_j) | \ \leq \ \Big(\sup_{s\in [0,1]} \sum_{j=1}^p |\alpha_j(s)|\Big) \Big( \sup_{\Vert f \Vert \leq 1} \sup_{t\in [0,1]} | f(t) | \Big) \\
&\leq & \sum_{j=1}^p \sup_{s\in [0,1]} |\alpha_j(s)| \ = \ \sum_{j=1}^p\Vert \alpha_j \Vert \ < \ 1.
\end{eqnarray*}
The operator norm satisfies $\Vert AB\Vert_{\mathcal{L}}\leq \Vert A\Vert_{\mathcal{L}}\Vert B\Vert_{\mathcal{L}}$, for $A,B$ operators in this space. Then the result follows from the corollary of Theorem~6.1 in \cite{bosq} since
$$\Vert \rho^{j_0} \Vert_{\mathcal{L}} \ \leq \ \Vert \rho \Vert_{\mathcal{L}}^{j_0} \ < \ 1,$$
for any finite positive $j_0$.
\end{proof}

The condition $\sum_{j=1}^p \Vert\alpha_j\Vert<1$ may be relaxed, as it is deduced from the proof, since it is enough to have $\Vert \rho^{j_0} \Vert_{\mathcal{L}}<1$. Proving this same result for a general kernel function $\phi$ is not straightforward, although possible. Each $\phi(s,\cdot)$ can be written as a pointwise limit of functions in $\calh_0(c_0)$, and then the operator $\rho$ of Equation \eqref{Eq:rho} would be defined as a limit of operators of finite rank. But the main issue is that one needs to impose the condition that this limit operator is bounded, which is usually unknown. Besides, additional smoothness conditions on the auto-covariance function are required.

\subsection*{Convergence when increasing the number of points}

In Section 5.1 of \cite{cambanis1985} the author studied the convergence of the finite dimensional approximations of the functions in $\mathcal{H}(X)$. Given a function $f\in\calh(X)$, let us denote the projection defined by the points in $T_p$ as 
$$f(\cdot)_{T_p} \ = \ \sum_{j=1}^p a_j(s)c_0(t_j,\cdot).$$
In Equation (5.3) of \cite{cambanis1985}, it is stated that if $f\in\mathcal{H}(X)$ and 
$$T_p^* = \arg\sup_{T_p\in D_p} f(T_p)'\Sigma_{T_p}f(T_p) = \arg\sup_{T_p\in D_p} \Vert f_{T_p} \Vert_\calh^2,$$ 
where $D_p = \{T_p\in[0,1]^p \ : \ t_1 < t_2 < \ldots < t_p\}$, then the square norm $\Vert f_{T_p^*} \Vert_\calh^2$ converges to $\Vert f \Vert_\calh^2$ as $p$ goes to infinity. In addition, if the points $T_p$ which generate the space onto which we project are not these optimal ones, the projected function $f_{T_p}$ could still converge to the real function $f$ under some conditions. For instance, in Equation (5.4) of the same book it is claimed that whenever $\boldsymbol{q}_p \in [0,1]^p$ form a regular sequence of points generated by a density $h$ (that is, $\boldsymbol{q}_p$ are the quantiles of $h$), then
$$\Vert f \Vert_\calh^2 - \Vert f_{\boldsymbol{q}_p} \Vert_\calh^2 \ = \ \Vert f - f_{\boldsymbol{q}_p} \Vert_\calh^2 \ \to \ 0 \ \text{ as } \ p\to\infty.$$
Note that we are projecting onto a closed subspace of $\mathcal{H}(X)$, so $\Vert f \Vert_\calh^2 = \Vert f_{\boldsymbol{q}_p} \Vert_\calh^2 + \Vert f - f_{\boldsymbol{q}_p} \Vert_\calh^2$. The author also give some convergence rates depending on how smooth the function $c_0$ is. 

Now, in the model that we have rewritten at the beginning of this appendix, the process $X_n$ is generated throughout a continuous function $\phi$ on $[0,1]^2$  such that $\phi(s,\cdot) \in \mathcal{H}(X)$ for all $s\in [0,1]$. We will denote as $T_p^*$ the points that minimize, for $T_p\in\Theta_p$
$$Q^1(T_p) \ = \ \int_0^1 \min_{\alpha_j(s)\in\mathbb{R}} \Big\Vert \phi(s,\cdot) - \sum_{j=1}^p \alpha_j(s)c_0(t_j,\cdot) \Big\Vert_\calh^2 \mathrm{d}s.$$
This criterion to select the points is equivalent to optimize $Q$ or $Q^0$ of Proposition \ref{Prop:Equiv} of the main document (as mentioned in the proof of that same proposition).
Then, using a similar notation for the projection of real random variables,
$$X_n(s)_{T_p} \ = \ \Psi_{X_{n-1}}^{-1}\left(\phi(s,\cdot)_{T_p}\right) \ = \ \sum_{j=1}^p \alpha_j(s) X_{n-1}(t_j),$$
we can prove the following result.

\begin{proposition}
If $X_n$ is given by Equation \eqref{fullmodel} of the main document, where $\phi$ is a continuous function in $[0,1]^2$ and $\varepsilon_n(s)$ is uncorrelated with $X_{n-1}(s)$ for every $s\in[0,1]$, then
$$\mathbb{E}\Big[\big\Vert \big(X_n(\cdot)-\varepsilon_n(\cdot)\big) - X_n(\cdot)_{T^p} \big\Vert_2^2\Big] \ \to \ 0 \ \text{ as } \ p\to\infty,$$
where $\Vert\cdot\Vert_2$ is the norm in $L^2[0,1]$.
\end{proposition}
\begin{proof}
We denote as $\boldsymbol{q}_p\in[0,1]^p$ the quantiles of a fixed density $h$. By Tonelli's Theorem we can change the order of the integrals 
\begin{eqnarray*}
\mathbb{E} \int_0^1 \big( (X_n(s)-\varepsilon_n(s)) - X_n(s)_{T_p} \big)^2 \mathrm{d}s &=&  \int_0^1 \mathbb{E} \big[ (X_n(s)-\varepsilon_n(s)) - X_n(s)_{T_p} \big]^2 \mathrm{d}s \\[1em]
&=& \int_0^1 \var\big( \big(X_n(s)-\varepsilon_n(s)\big) - X_n(s)_{T_p} \big) \mathrm{d}s \\[1em]
&=& Q_1(T_p^*) - \mathbb{E}\big[\Vert \varepsilon_n \Vert_2^2\big] \ = \ Q^1(T_p^*) \\[1em]
&\leq& \ Q^1(\boldsymbol{q}_p) \ = \ \int_0^1 \big\Vert \phi(s,\cdot) - \phi(s,\cdot)_{\boldsymbol{q}_p} \big\Vert_\calh^2 \mathrm{d}s.
\end{eqnarray*}
That is, we have to see that the sequence of functions $g_p(s) = \Vert \phi(s,\cdot) - \phi(s,\cdot)_{\boldsymbol{q}_p} \Vert_\calh^2$ converges to zero in $L^1[0,1]$.
Since the points $\boldsymbol{q}_p$ are a regular sequence, we know that $g_p(s) \to 0$ as $p\to\infty$ for each $s\in[0,1]$ (Equation (5.4) of \cite{cambanis1985}). Besides, these functions are bounded,
$$|g_p(s)| \ = \ g_p(s) \ = \ \Vert \phi(s,\cdot)\Vert_\calh^2 - \Vert\phi(s,\cdot)_{\boldsymbol{q}_p} \Vert_\calh^2 \ \leq \ \Vert \phi(s,\cdot)\Vert_\calh^2,$$
and since $\int_0^1 \Vert \phi(s,\cdot)\Vert_\calh^2 \mathrm{d}s\leq \sup_s \Vert \phi(s,\cdot)\Vert_\calh^2 <\infty$ (because $s\mapsto \Vert \phi(s,\cdot)\Vert_\calh^2$ is the composition of two continuous functions $s\mapsto \phi(s,\cdot)$ and $f\mapsto\Vert f \Vert_\calh^2$ over the compact $[0,1]$) we obtain the result using the dominated convergence theorem. 
\end{proof}


\newpage
\phantomsection
\addcontentsline{toc}{section}{Appendix B}
{\large \textbf{APPENDIX B: Additional proofs}}

\

In this appendix we include the proofs that are mainly based on the pointwise results of \cite{berrendero2018}.

\

\textbf{Proof of Proposition \ref{Prop:Equiv}}

This proof is inspired by the one of Proposition 1 of \cite{berrendero2018}. Pointwise for each $s\in[0,1]$, using Lo\`eve's isometry we have that 
$$\var\Big( X_n(s) - \sum_{j=1}^p \alpha_j(s)X_{n-1}(t_j) \Big) \ = \ \big\Vert \phi(s,\cdot) - \sum_{j=1}^p \alpha_j(s)c_0(t_j,\cdot) \big\Vert_\calh^2 + \sigma(s),$$
where $\sigma(s) = \var(\varepsilon_n(s))\leq \Vert\E [\varepsilon_n^2] \Vert \leq \E \Vert \varepsilon_n^2\Vert <\infty$, so the minimizing values $\alpha_j(s)$ are the same for both sides of the equality. Again pointwise for each $s\in[0,1]$, using the reproducing property of $\calh(X)$,
$$\Big\Vert \phi(s,\cdot) - \sum_{j=1}^p \alpha_j(s)c_0(t_j,\cdot) \Big\Vert_{\calh}^2 = \Vert\phi(s,\cdot)\Vert_{\calh} + \sum_{i,j=1}^p \alpha_i(s)\alpha_j(s) c_0(t_i,t_j)  - 2 \sum_{j=1}^p \alpha_j(s)\phi(s,t_j).$$
Since $c_0$ is a positive-definite function, this last function is convex in $\alpha_j(s)$ for each $s\in[0,1]$. Therefore we can compute its minimum pointwisely, which is achieved at $(\alpha_1^*(\cdot),\ldots,\alpha_p^*(\cdot))' = \Sigma_{T_p}^{-1} c_1(\cdot,T_p)$, since $c_1(s,t) = \phi(s,t)$ for each $s$ (Equation \eqref{Eq:phiC1} of Appendix A). Then if we substitute this optimum in the previous equation we get
$$\min_{\alpha_j(s)\in \mathbb{R}} \ \big\Vert \phi(s,\cdot) - \sum_{j=1}^p \alpha_j(s)c_0(t_j,\cdot) \big\Vert_\calh^2 \ = \ \Vert\phi(s,\cdot)\Vert_\calh^2 - c_1(s,T_p)'\Sigma_{T_p}^{-1} c_1(s,T_p).$$

Hence, integrating over $s\in [0,1]$, 
$$Q(T_p) = \int_0^1 \sigma(s) \dd s + \int_0^1 \Vert \phi(s,\cdot) \Vert_{\calh}^2 \dd s - Q^0(T_p) = C - Q^0(T_p).$$ 
This constant $C$ is finite since the integral of $\sigma(s)$ is bounded by $\Vert\E [\varepsilon_n^2]\Vert<\infty$ and 
$\int_0^1 \Vert \phi(s,\cdot) \Vert_{\calh}^2 \dd s \leq \sup_{s\in[0,1]} \Vert \phi(s,\cdot) \Vert_{\calh}^2$
being $\Vert \phi(s,\cdot) \Vert_\calh^2$ a continuous function on $[0,1]$ (it is the composition of two continuous functions, $s\mapsto \phi(s,\cdot) = c_1(s,\cdot)$ and $f\mapsto \Vert f \Vert_{\calh}$).

\


\textbf{Proof of Lemma \ref{Lemma:ContConv}}\label{Sec:ProofContConv}

Denoting
$$Q^0(T_p) \ = \ \int_0^1 c_1(s,T_p)' \Sigma_{T_p}^{-1} c_1(s,T_p) \ \mathrm{d}s \ = \ \int_0^1 q^0(T_p;s) \mathrm{d}s,$$
$$\widehat{Q}^0_m(T_p) \ = \ \int_0^1 \widehat{c}_1(s,T_p)' \widehat{\Sigma}_{T_p}^{-1} \widehat{c}_1(s,T_p) \ \mathrm{d}s \ = \ \int_0^1 \widehat{q}^0_m(T_p;s) \mathrm{d}s,$$
we can extend the proofs of Lemmas 2 and 3 of \cite{berrendero2018} to our setting.

For the continuity of the functions, it can be shown that $q^0(T_p;s)$ and $\wh q_m^0(T_p;s)$ are continuous in $T_p$ for each $s\in [0,1]$, using the same reasoning as in the proof of Lemma 2 of \cite{berrendero2018} but using now Lemma \ref{Lemma:UnifConv}. Then if $\Vert T_p - S_p\Vert_{\R^p} < \delta$, where $\Vert\cdot\Vert_{\R^p}$ is the usual vector norm,
$$\vert Q^0(T_p) - Q^0(S_p) \vert \ = \ \Big\vert \int_0^1 \big(q^0(T_p;s) - q^0(S_p;s)\big)\mathrm{d}s \Big\vert \ \leq \ \int_0^1 \vert q^0(T_p;s) - q^0(S_p;s)\vert\mathrm{d}s \ < \ \varepsilon.$$
And equivalently to see that $\widehat{Q}^0_m$ is continuous with probability one.

The uniform convergence is straightforward in view of the definition of $Q^0$. By Lemma \ref{Lemma:UnifConv}, the vector $\wh c_1(s,T_p)$ converge uniformly a.s. to $c_1(s,T_p)$. Besides, by the same reasoning as in Lemma 3 of \cite{berrendero2018}, using Lemma \ref{Lemma:UnifConv} now, one gets the uniform convergence of the inverse covariance matrices. Then, we get
$$\sup_{(s,T_p)\in [0,1]\times\Theta_p} \big\vert \widehat{q}^0_m(T_p;s) - q^0(T_p;s)\big\vert \ \stas \ 0,$$
which implies the convergence of the integral over $s\in [0,1]$.

\


\textbf{Proof of Theorem \ref{Theo:Ts}}\label{Sec:ProofTs}

(a) Because of the equivalence of the criteria proved in Proposition \ref{Prop:Equiv}, it is enough to see that $T^*$ is the only global minimum of $Q(T_{p^*})$ in $\Theta_{p^*}$. From Equation \eqref{Q1} it is clear that $T^*$ minimizes $Q$ since
$$Q(T_{p^*}) = \int_0^1 \var\big( X_n(s) - X_n(s)_{T_{p^*}} \big) \ \mathrm{d}s = \int_0^1 \var\big( X_n(s)_{T^*} - X_n(s)_{T_{p^*}} \big) \ \mathrm{d}s + \Vert \var(\varepsilon_n)\Vert_2,$$
where $\Vert\cdot\Vert_2$ is the norm in $L^2[0,1]$, the space of square integrable functions over $[0,1]$. Therefore, its minimum value is $\Vert \var(\varepsilon_n)\Vert_2 = \Vert \sigma\Vert_2$. If there exists another vector $S^*\neq T^*$ which also achieves this value, it must be $\var\big( X_n(s)_{T^*} - X_n(s)_{S^*} \big)=0$ for almost every $s\in [0,1]$ (except on a set of measure zero with regard to the Lebesgue measure). However, it is enough to have one point $s_0$ in which this equality holds. For this point we have that $X_n(s_0)_{T^*}=X_n(s_0)_{S^*}$ a.s., which contradicts the assumption that the covariance matrix $\Sigma_{T^*\cup S^*}$ is invertible, and then $T^*=S^*$. 

(b) The result follows from the fact that $\widehat{Q}^0_m$ and $Q^0$ are continuous functions such that $\widehat{Q}^0_m$ tends uniformly a.s. to $Q^0$ (Lemma \ref{Lemma:ContConv}) and $Q^0$ has a unique maximum in $\Theta_{p^*}$ (part (a)).

(c) The same argument as in the proof of Theorem 1.c of \cite{berrendero2018}.

\


\textbf{Proof of Theorem \ref{Theo:Resp}}\label{Sec:ProofResp}

(a) We derive a proof similar to the proof of Theorem 2 of \cite{berrendero2018}. We can decompose the norm in the statement as
\begin{align*}
\hspace*{-3mm}\big\Vert \wh X_{n}(\cdot)_{\widehat{T}_{p^*,m}} \hspace*{-3mm}- X_n(\cdot)_{T^*} \big\Vert &= \big\Vert \wh c_1(\cdot,\wh T_{p^*,m})' \, \widehat{\Sigma}_{\wh T_{p^*,m}}^{-1}X_{n-1}({\wh T_{p^*,m}}) - c_1(\cdot,T^*)' \, \Sigma_{T^*}^{-1} X_{n-1}(T^*)\big\Vert \nonumber\\
&\leq \big\Vert \big(\wh c_1(\cdot,\wh T_{p^*,m}) - c_1(\cdot,T^*)\big)' \, \widehat{\Sigma}_{\wh T_{p^*,m}}^{-1}X_{n-1}({\wh T_{p^*,m}}) \big\Vert \nonumber\\
& + \ \big\Vert c_1(\cdot,T^*)' \, \big( \widehat{\Sigma}_{\wh T_{p^*,m}}^{-1}\hspace*{-2mm}X_{n-1}({\wh T_{p^*,m}}) - \Sigma_{T^*}^{-1} X_{n-1}(T^*) \big) \big\Vert \nonumber\\
&\leq \big\Vert \Sigma_{T^*}^{-1} X_{n-1}(T^*) + \epsilon \, \big\Vert_{\R^{p^*}} \sup_{s\in[0,1]} \big\Vert \wh c_1(s,\wh T_{p^*,m}) - c_1(s,T^*) \big\Vert_{\R^{p^*}} \tag{A}\label{Eq:A}\\
& + \ \big\Vert \widehat{\Sigma}_{\wh T_{p^*,m}}^{-1} \hspace*{-2mm}X_{n-1}({\wh T_{p^*,m}}) - \Sigma_{T^*}^{-1} X_{n-1}(T^*) \big\Vert_{\R^{p^*}} \sup_{s\in[0,1]} \big\Vert c_1(s,T^*)\big\Vert_{\R^{p^*}}, \tag{B}\label{Eq:B}
\end{align*}
where $\Vert\cdot\Vert_{\R^{p^*}}$ denotes the Euclidean norm in $\R^{p^*}$. For the term in Equation \eqref{Eq:A}, it is enough to see the convergence of each entry of the covariance vector. For any $1\leq j \leq p^*$,
\begin{eqnarray*}
\sup_{s\in[0,1]}\big\vert\wh c_1(s,\wh t_j) - c_1(s,t_j^*) \big\vert &\leq& \sup_{s\in[0,1]}\big\vert\wh c_1(s,\wh t_j) - c_1(s,\wh t_j) \big\vert + \sup_{s\in[0,1]}\big\vert  c_1(s,\wh t_j) - c_1(s,t_j^*) \big\vert \\
&\leq& \sup_{s,t\in[0,1]} \big\vert\wh c_1(s,t) - c_1(s,t) \big\vert + \sup_{s\in[0,1]}\big\vert  c_1(s, \,\wh t_j-t_j^*) \big\vert.
\end{eqnarray*}
The first addend goes a.s. to zero by Lemma 1, and the second addend due to the a.s. convergence of $\wh T_{p^*,m}$ to $T^*$ and the continuity of both $c_1$ and the norm in $C[0,1]$ (continuous mapping theorem).

The first term in Equation \eqref{Eq:B} is bounded by 
\begin{eqnarray*}
\big\Vert \widehat{\Sigma}_{\wh T_{p^*,m}}^{-1} \hspace*{-2mm}X_{n-1}({\wh T_{p^*,m}}) - \Sigma_{T^*}^{-1} X_{n-1}(T^*) \big\Vert_{\R^{p^*}} &\leq& \big\Vert \widehat{\Sigma}_{\wh T_{p^*,m}}^{-1} \hspace*{-2mm}X_{n-1}({\wh T_{p^*,m}}) - \Sigma_{\wh T_{p^*,m}}^{-1} X_{n-1}(\wh T_{p^*,m}) \big\Vert_{\R^{p^*}} \\
&& + \ \big\Vert \Sigma_{\wh T_{p^*,m}}^{-1} \hspace*{-2mm}X_{n-1}({\wh T_{p^*,m}}) - \Sigma_{T^*}^{-1} X_{n-1}(T^*) \big\Vert_{\R^{p^*}}.
\end{eqnarray*}
The first term is bounded by $\sup_{T\in\Theta_{p^*}} \big\Vert \big( \widehat{\Sigma}_T^{-1} - \Sigma_{T}^{-1}\big) X_{n-1}(T) \big\Vert_{\R^{p^*}}$. Using Lemma \ref{Lemma:UnifConv} we show that the sample version of the inverse of the covariance matrix, ${\wh \Sigma}_{T}^{-1}$, as a function on $\Theta_{p^*}$, converges uniformly with probability one to its population counterpart $\Sigma_{T}^{-1}$ (using a similar reasoning as in the proof of Lemma 3 of \cite{berrendero2018}). Since $X_{n-1}$ has continuous trajectories on a compact space, the product goes almost sure to zero. The second addend also goes to zero due to Theorem 1(b) and the continuity of the involved functions. 

(b) The statement is equivalent to see that the real valued random variables
$Z_m \ = \ \big\Vert \widehat{X}_n(\cdot)_{\widehat{T}_{p^*,m}} - X_n(\cdot)_{T^*} \big\Vert^2$ converge to 0 in the $L^1$-norm. From part a) we know that they converge a.s. to zero, so it only remains to check that the sequence $Z_m$ is uniformly integrable (Vitali's convergence theorem). In order to do this, it is enough to check that $\sup_m \E[Z_m^\eta]<\infty$ for some $\eta>1$ (de la Vall\'ee-Poussin's theorem). Since the function $f(y)=y^{2\eta}$ is convex in this case,
$$Z_m^\eta \ \leq \ \big(\big\Vert \widehat{X}_n(\cdot)_{\widehat{T}_{p^*,m}} \big\Vert + \big\Vert X_n(\cdot)_{T^*} \big\Vert\big)^{2\eta} \ \leq \ 0.5\big[\big(2\big\Vert \widehat{X}_n(\cdot)_{\widehat{T}_{p^*,m}} \big\Vert\big)^{2\eta} + \big(2\big\Vert X_n(\cdot)_{T^*} \big\Vert\big)^{2\eta}\big].$$
Thus, we have to verify that
$$2^{2\eta-1} \Big(\sup_m\E\big[ \big\Vert \widehat{X}_n(\cdot)_{\widehat{T}_{p^*,m}} \big\Vert^{2\eta} \big] + \E\big[\big\Vert X_n(\cdot)_{T^*} \big\Vert^{2\eta}\big] \Big)<\infty.$$
Let us start with the first addend. Using Lemma \ref{Lemma:UnifConv} the supremum of $\Vert \widehat{c}_{1}(s,T_{p^*})'\widehat{\Sigma}_{T_{p^*}}^{-1} - c_{1}(s,T_{p^*})'\Sigma_{T_{p^*}}^{-1} \Vert_{\R^{p^*}}$ for $(s,T_{p^*})\in [0,1]\times \Theta_{p^*}$ goes a.s to zero, then we bound the norm as
\begin{eqnarray*}
\big\Vert \widehat{X}_n(\cdot)_{\widehat{T}_{p^*,m}} \big\Vert^{2\eta} &=& \Big(\sup_{s\in [0,1]} \big\vert \widehat{c}_{1}(s,\widehat{T}_{p^*,m})'\,\widehat{\Sigma}_{\widehat{T}_{p^*,m}}^{-1}X_{n-1}(\widehat{T}_{p^*,m})\big\vert \Big)^{2\eta} \\[1em]
& \leq & \big\Vert X_{n-1}(\widehat{T}_{p^*,m}) \big\Vert_{\R^{p^*}}^{2\eta} \Big( \sup_{s\in [0,1]} \big\Vert \widehat{c}_{1}(s,\widehat{T}_{p^*,m})'\widehat{\Sigma}_{\widehat{T}_{p^*,m}}^{-1} \big\Vert_{\R^{p^*}} \Big)^{2\eta} \\[1em]
&\leq& \Vert X_{n-1}(\widehat{T}_{p^*,m}) \Vert_{\R^{p^*}}^{2\eta} \Big( \sup_{s\in[0,1],T_{p^*}\in\Theta_{p^*}} \big\Vert c_{1}(s,T_{p^*})'\,\Sigma_{T_{p^*}}^{-1}\big\Vert_{\R^{p^*}} + \epsilon \Big)^{2\eta} \\[1em]
& \leq & C \ \Vert X_{n-1}(\widehat{T}_{p^*,m}) \Vert_{\R^{p^*}}^{2\eta},
\end{eqnarray*}
where $C<\infty$, since the function involved in the supremum is a continuous function on the compact space $[0,1]\times\Theta_{p^*}$. To conclude we can use the same reasoning as in the proof of Theorem 2 of \cite{berrendero2018}.

For the second addend, which does not depend on the sample size,
\begin{eqnarray*}
\E \Big[ \Big( \sup_{s\in[0,1]} \big\vert X_n(s)_{T^*}\big\vert \Big)^{2\eta} \Big] &=& \E \Big[ \Big( \sup_{s\in[0,1]} \big\vert c_1(s,T^*)'\,\Sigma_{T^*}^{-1} X_{n-1}(T^*) \big\vert \Big)^{2\eta} \Big] \\
&\leq& \Big( \sup_{s\in[0,1]} \big\Vert c_1(s,T^*)'\,\Sigma_{T^*}^{-1} \big\Vert_{\R^{p^*}} \Big)^{2\eta} \ \E \Big[ \big\Vert X_{n-1}(T^*) \big\Vert_{\R^{p^*}}^{2\eta} \Big],
\end{eqnarray*}
where the expectation is finite by hypothesis.


\newpage
\phantomsection
\addcontentsline{toc}{section}{Appendix C}
{\large \textbf{APPENDIX C: Pseudocodes and tables}}

\vspace*{1cm}

\begin{algorithm}
\caption{Variable selection for a given $p$}\label{Al:VS}
\begin{algorithmic}[1]
\Procedure{RKHS variable selection}{}
\BState \emph{First point}:
\State $Quot(grid) \gets \widehat c_0(t,t)^{-1} \int_0^1 \widehat c_1(s,t)^2 \mathrm{d} s, \ \forall t\in grid$
\State $M(1) \gets max(Quot)$
\State $pts(1) \gets \text{which } Quot == M(1)$

\BState \emph{Rest of the points}:
\For {i in 2 to p}
    \State $Quot(grid \backslash pts) \gets \text{Quotient of Eq. } \eqref{Eq:QnIterM} \ \forall t_{p+1}\in \{grid \text{ and not in } pts\}$
    \State $M(i) \gets max(Quot)$
    \State $pts(i) \gets \text{which } Quot == M(i)$
\EndFor
\BState \emph{Return pts (points) and M}
\EndProcedure
\end{algorithmic}
\end{algorithm}

\

\begin{algorithm}
\caption{Cluster approximation to estimate $\widehat p$}\label{Al:Points}
\begin{algorithmic}[1]
\Procedure{Number of points}{}
\State $M, points \gets \text{ RKHS method for P points }$
\State $L_m \gets log(M)$
\BState \emph{k-means procedure with} $k=2$:
\State $clusters \gets kmeans(L_m)$
\State $cl1 \gets clusters(1)$
\State $\widehat p \gets \text{ tail of } clusters==cl1$
\EndProcedure
\end{algorithmic}
\end{algorithm}

\begin{table}[ht]
\small
\centering
\resizebox{\textwidth}{!}{
\begin{tabular}{lcll|cccccc}
  \hline
 &&&& RKHS+cl & RKHS+CV & fFPE & KR & Exact & Naive \\ 
  \hline
  \multirow{8}{1.5cm}{Sparse with log.} & \multirow{4}{0.5cm}{$\varepsilon_1$} & \multirow{2}{*}{Disc.}& L2 & \textbf{0.64} & 0.71 & 0.65 & 0.66 & 0.55 & 3.24 \\
&&&  sup & \textbf{0.68} & 0.71 & \textbf{0.68} & 0.82 & 0.65 & 1.64 \\
&& \multirow{2}{*}{Funct.}& L2 & 0.65 & \textbf{0.64} & 0.67 & 0.67 & 0.56 & 3.32 \\ 
&&& sup & \textbf{0.65} & \textbf{0.65} & \textbf{0.65} & 0.83 & 0.62 & 1.67 \\ 
& \multirow{4}{0.7cm}{$\varepsilon_2$} & \multirow{2}{*}{Disc.}& L2 & \textbf{0.16} & 0.18 & 0.17 & 0.18 & 0.14 & 1.32 \\ 
&&&  sup & \textbf{0.30} & 0.32 & \textbf{0.30} & 0.39 & 0.29 & 0.80 \\
&& \multirow{2}{*}{Funct.}& L2 & \textbf{0.32} & \textbf{0.32} & 0.33 & 0.34 & 0.28 & 2.63 \\ 
&&&  sup & \textbf{0.55} & \textbf{0.55} & 0.56 & 0.77 & 0.53 & 1.57 \\
   \hline
\multirow{8}{1.5cm}{Sparse with sins} & \multirow{4}{0.5cm}{$\varepsilon_1$} &  \multirow{2}{*}{Disc.}&   L2 & \textbf{0.72} & 0.74 & 0.83 & 0.94 & 0.60 & 2.42 \\ 
&&&  sup & \textbf{0.71} & 0.73 & 0.84 & 0.95 & 0.67 & 1.50 \\ 
&& \multirow{2}{*}{Funct.}& L2 & \textbf{0.78} & 0.81 & 0.81 & 0.94 & 0.65 & 2.60 \\
&&&  sup & \textbf{0.77} & \textbf{0.77} & \textbf{0.77} & 0.95 & 0.69 & 1.53 \\
& \multirow{4}{0.7cm}{$\varepsilon_2$} & \multirow{2}{*}{Disc.}& L2 & \textbf{0.38} & \textbf{0.38} & 0.42 & 0.48 & 0.33 & 1.10 \\
&&&  sup & \textbf{0.36} & 0.37 & 0.42 & 0.49 & 0.34 & 0.75 \\
&& \multirow{2}{*}{Funct.}& L2 & \textbf{0.78} & 0.79 & 0.80 & 0.91 & 0.69 & 2.19 \\
&&& sup & \textbf{0.75} & \textbf{0.75} & \textbf{0.75} & 0.94 & 0.68 & 1.47 \\
   \hline
\multirow{8}{1.5cm}{O.U.} & \multirow{4}{0.5cm}{$\varepsilon_1$} & \multirow{2}{*}{Disc.}& L2 & 1.07 & 1.01 & \textbf{1.00} & 1.15 & 0.83 & 2.33 \\
&&&  sup & \textbf{0.88} & \textbf{0.88} & 0.93 & 0.98 & 0.88 & 1.35 \\
&& \multirow{2}{*}{Funct.}& L2 & \textbf{1.00} & \textbf{1.00} & 1.05 & 1.20 & 0.85 & 2.49 \\
&&&  sup & \textbf{0.91} & \textbf{0.91} & 0.92 & 0.98 & 0.86 & 1.34 \\
& \multirow{4}{0.5cm}{$\varepsilon_2$} & \multirow{2}{*}{Disc.}& L2 & \textbf{0.31} & \textbf{0.31} & 0.35 & 0.42 & 0.30 & 0.65 \\
&&& sup & \textbf{0.44} & \textbf{0.44} & 0.47 & 0.50 & 0.44 & 0.65 \\
&& \multirow{2}{*}{Funct.}& L2 & \textbf{0.66} & \textbf{0.66} & 0.68 & 0.83 & 0.59 & 1.26 \\
&&& sup & \textbf{0.85} & \textbf{0.85} & 0.86 & 0.94 & 0.81 & 1.21 \\
\hline
\multirow{4}{1.5cm}{FAR} & \multirow{2}{0.5cm}{$\varepsilon_1$} && L2 & \textbf{1.00} & 1.01 & 1.01 & 1.13 & 0.85 & 2.20 \\
&&&   sup & \textbf{1.00} & \textbf{1.00} & \textbf{1.00} & 1.04 & 0.91 & 1.45 \\
& \multirow{2}{0.5cm}{$\varepsilon_2$} && L2 & \textbf{0.96} & \textbf{0.96} & 0.97 & 1.08 & 0.81 & 1.96 \\
&&& sup & \textbf{0.98} & \textbf{0.98} & \textbf{0.98} & 1.02 & 0.89 & 1.39 \\
\hline
\end{tabular}
}
\caption{Errors for the simulated data sets ($e_1$ and $e_2$ errors of Eq. \eqref{Eq:Error})} 
\label{Table:sim}
\end{table}

\begin{table}[ht]
\small
\centering
\begin{tabular}{lcll|ccccc}
  \hline
 &&&& RKHS+cl & RKHS+CV & fFPE & KR & Naive \\ 
  \hline 
\multirow{8}{1.5cm}{PM10} & \multirow{4}{1.5cm}{$\varepsilon_1$ error} & \multirow{2}{*}{Disc.}& L2 & 0.97 & \textbf{0.74} & 0.82 & 1.48 & 1.65 \\
&&&  sup & 0.92 & \textbf{0.86} & \textbf{0.86} & 1.06 & 1.15 \\
&& \multirow{2}{*}{Func.} & L2 & \textbf{0.70} & 0.72 & 0.82 & 1.59 & 1.68 \\
&&&   sup & \textbf{0.84} & \textbf{0.84} & 0.85 & 1.08 & 1.11 \\
& \multirow{4}{1.5cm}{$\varepsilon_2$ error} & \multirow{2}{*}{Disc.}& L2 & 0.56 & \textbf{0.47} & 0.50 & 0.86 & 0.80 \\
&&&  sup & 0.85 & 0.81 & \textbf{0.80} & 0.98 & 1.02 \\
&& \multirow{2}{*}{Func.} & L2 & \textbf{0.48} & \textbf{0.48} & \textbf{0.48} & 0.85 & 0.76 \\
&&& sup & 0.79 & 0.79 & \textbf{0.78} & 0.99 & 0.97 \\ 
   \hline 
\multirow{8}{1.5cm}{Traffic} & \multirow{4}{1.5cm}{$\varepsilon_1$ error} & \multirow{2}{*}{Disc.}& L2 & \textbf{1.00} & 1.02 & \textbf{1.00} & 1.04 & 67.48 \\ 
&&&  sup & 1.02 & \textbf{1.01} & \textbf{1.01} & \textbf{1.01} & 4.70 \\
&& \multirow{2}{*}{Func.} & L2 & 1.37 & \textbf{1.21} & 1.49 & 1.42 & 240.57 \\
&&&   sup & 1.04 & \textbf{0.98} & 1.05 & 1.11 & 10.22 \\
& \multirow{4}{1.5cm}{$\varepsilon_2$ error} & \multirow{2}{*}{Disc.}& L2 & 0.90 & 0.89 & \textbf{0.83} & 0.95 & 40.57 \\
&&&  sup & 0.99 & \textbf{0.98} & \textbf{0.98} & 0.99 & 4.26 \\
&& \multirow{2}{*}{Func.} & L2 & 0.83 & 0.80 & \textbf{0.75} & 0.92 & 62.07 \\
&&& sup & 0.91 & \textbf{0.89} & 0.90 & 1.00 & 6.82 \\
   \hline 
   \multirow{8}{1.5cm}{Temp} & \multirow{4}{1.5cm}{$\varepsilon_1$ error} & \multirow{2}{*}{Disc.}& L2 & \textbf{0.45} & 0.53 & 5.28 & 1.11 & 37.78 \\ 
&&&  sup & \textbf{0.66} & 0.67 & 2.10 & 1.08 & 3.55 \\
&& \multirow{2}{*}{Func.} & L2 & \textbf{0.63} & 0.82 & 5.30 & 1.12 & 38.20 \\
&&&   sup & \textbf{0.66} & 0.72 & 2.11 & 1.07 & 3.54 \\
& \multirow{4}{1.5cm}{$\varepsilon_2$ error} & \multirow{2}{*}{Disc.}& L2 & 0.25 & \textbf{0.24} & 2.87 & 1.03 & 14.25 \\
&&&  sup & 0.59 & \textbf{0.54} & 1.95 & 1.05 & 2.98 \\
&& \multirow{2}{*}{Func.} & L2 & \textbf{0.37} & \textbf{0.37} & 2.85 & 1.03 & 14.23 \\
&&& sup & \textbf{0.58} & 0.60 & 1.96 & 1.04 & 2.98 \\
   \hline 
   \multirow{8}{1.5cm}{Utility} & \multirow{4}{1.5cm}{$\varepsilon_1$ error} & \multirow{2}{*}{Disc.}& L2 & 0.11 & \textbf{0.09} & 0.10 & 1.16 & 18.72 \\ 
&&&  sup & 0.35 & \textbf{0.34} & \textbf{0.34} & 1.02 & 3.33 \\
&& \multirow{2}{*}{Func.} & L2 & 0.09 & \textbf{0.08} & 0.09 & 1.17 & 18.94 \\
&&&   sup & \textbf{0.29} & 0.30 & 0.31 & 1.01 & 3.37 \\
& \multirow{4}{1.5cm}{$\varepsilon_2$ error} & \multirow{2}{*}{Disc.}& L2 & 0.08 & \textbf{0.07} & \textbf{0.07} & 0.93 & 15.08 \\
&&&  sup & 0.33 & \textbf{0.32} & \textbf{0.32} & 0.98 & 3.19 \\
&& \multirow{2}{*}{Func.} & L2 & \textbf{0.06} & \textbf{0.06} & \textbf{0.06} & 0.92 & 15.16 \\
&&& sup & \textbf{0.28} & \textbf{0.28} & 0.30 & 0.98 & 3.24 \\
   \hline
\end{tabular}
\caption{Errors for real data sets ($e_1$ and $e_2$ of Eq. \eqref{Eq:Error})} 
\label{Table:real}
\end{table}

\begin{table}[ht]
\small
\centering
\begin{tabular}{lcccccc}
  \hline
&RKHS+cl & RKHS+CV & RKHS+cl & RKHS+CV & \multirow{2}{*}{fFPE} & \multirow{2}{*}{KR} \\
& (disc) &  (disc) & (funct) & (funct) && \\ 
  \hline
PM10 & \textbf{0.09} & \textbf{0.09} & 0.89 & 1.08 & 0.71 & 2.62 \\ 
Traffic & 2.13 & 2.14 & 32.66 & 32.06 & \textbf{0.52} & 1.04 \\
Temp & 0.27 & \textbf{0.20} & 3.49 & 3.80 & 0.37 & 0.51 \\
Utility & \textbf{0.23} & \textbf{0.23} & 4.38 & 4.05 & 0.56 & 1.02 \\
   \hline
\end{tabular}
\caption{Execution times (secs) for the real data sets.} 
\label{Table:timesReal}
\end{table}

\begin{table}[ht]
\small
\centering
\begin{tabular}{ll|cccccc}
  \hline
&\multirow{2}{*}{n} & RKHS+cl & RKHS+CV & RKHS+cl & RKHS+CV & \multirow{2}{*}{fFPE} & \multirow{2}{*}{KR} \\
& & (disc) &  (disc) & (funct) & (funct) && \\ 
  \hline 
\multirow{5}{1.5cm}{Sparse with log.} & 50 & \textbf{0.05} & \textbf{0.05} & 0.34 & 0.37 & 0.57 & 1.50 \\ 
&  100 & \textbf{0.05} & 0.06 & 0.34 & 0.39 & 0.63 & 2.60 \\ 
&  150 & \textbf{0.04} & 0.05 & 0.33 & 0.40 & 0.84 & 3.79 \\ 
&  200 & \textbf{0.05} & \textbf{0.05} & 0.33 & 0.39 & 1.00 & 5.09 \\ 
&  250 & \textbf{0.05} & 0.06 & 0.34 & 0.38 & 1.24 & 6.37 \\ 
   \hline 
\multirow{5}{1.5cm}{Sparse with sins} & 50 & 0.06 & \textbf{0.05} & 0.85 & 0.98 & 0.62 & 2.26 \\ 
&  100 & \textbf{0.05} & 0.06 & 0.81 & 1.03 & 0.78 & 4.07 \\ 
&  150 & \textbf{0.05} & 0.06 & 0.84 & 1.05 & 0.96 & 5.85 \\ 
&  200 & \textbf{0.05} & 0.06 & 0.86 & 1.06 & 1.22 & 7.94 \\ 
&  250 & \textbf{0.05} & 0.06 & 0.86 & 1.07 & 1.44 & 9.96 \\ 
   \hline 
\multirow{5}{1.5cm}{O.U} & 50 & 0.26 & \textbf{0.23} & 5.07 & 5.36 & 0.55 & 1.47 \\ 
&  100 & 0.26 & \textbf{0.25} & 5.04 & 5.25 & 0.72 & 2.59 \\ 
&  150 & 0.28 & \textbf{0.26} & 5.14 & 5.51 & 0.93 & 3.73 \\ 
&  200 & \textbf{0.27} & \textbf{0.27} & 5.07 & 5.34 & 1.24 & 5.05 \\ 
&  250 & 0.27 & \textbf{0.25} & 4.94 & 5.23 & 1.58 & 6.30 \\ 
   \hline 
\multirow{5}{1.5cm}{FAR} & 50 & \textbf{0.05} & \textbf{0.05} & 0.53 & 0.59 & 0.51 & 1.71 \\ 
&  100 & \textbf{0.05} & \textbf{0.05} & 0.54 & 0.59 & 0.72 & 3.05 \\ 
&  150 & \textbf{0.05} & \textbf{0.05} & 0.53 & 0.63 & 1.03 & 4.37 \\ 
&  200 & \textbf{0.05} & \textbf{0.05} & 0.54 & 0.63 & 1.35 & 5.88 \\ 
&  250 & \textbf{0.05} & 0.06 & 0.54 & 0.63 & 1.63 & 7.40 \\ 
   \hline
\end{tabular}
\caption{Execution times (secs) for the simulated data sets.} 
\label{Table:timesSim}
\end{table}

\begin{table}[ht]
\small
\centering
\begin{tabular}{l|cc}
  \hline
q & RKHS+cl (funct) & RKHS+CV (funct) \\ 
  \hline
1 & 0.36 & 0.38 \\ 
  2 & 1.71 & 1.82 \\ 
  3 & 5.04 & 5.16 \\ 
  4 & 8.86 & 9.07 \\ 
  5 & 13.92 & 14.25 \\ 
   \hline
\end{tabular}
\caption{Execution times (secs) for FCAR(q) with functional implementation.} 
\label{Table:timesFCAR}
\end{table}

\end{document}